\documentclass[11pt,twoside]{article}
\usepackage[margin=1in]{geometry}

\usepackage{jeffe,graphicx}		
\usepackage{microtype}
\DeclareGraphicsExtensions{.pdf,.png,.jpg}

\hypersetup{colorlinks=true,urlcolor=Black,citecolor=Black,linkcolor=Black}

\usepackage[mathletters]{ucs}		
\usepackage[utf8x]{inputenc}

\usepackage[charter]{mathdesign}	
\usepackage{berasans,beramono}
\SetMathAlphabet{\mathsf}{bold}{\encodingdefault}{\sfdefault}{b}{\updefault}
\SetMathAlphabet{\mathtt}{bold}{\encodingdefault}{\ttdefault}{b}{\updefault}
\SetMathAlphabet{\mathsf}{normal}{\encodingdefault}{\sfdefault}{\mddefault}{\updefault}
\SetMathAlphabet{\mathtt}{normal}{\encodingdefault}{\ttdefault}{\mddefault}{\updefault}
\RequirePackage{microtype}
\usepackage{eucal,textcomp}

\usepackage{cite}		
\usepackage{enumerate}

\RequirePackage{colortbl} 
\RequirePackage{array}
\RequirePackage{arydshln}
\newtheorem{theorem}{Theorem}[section]
\newtheorem{lemma}[theorem]{Lemma}
\newtheorem{corollary}[theorem]{Corollary}
\numberwithin{figure}{section}

\def\Wind{\operatorname{\mathit{wind}}}
\def\Diam{\operatorname{\mathit{diam}}}

\usepackage{mathtools}

\begin{document}

\begin{titlepage}

\urldef\paperURL\url{http://jeffe.cs.illinois.edu/pubs/defect.html}

\title{Electrical Reduction, Homotopy Moves, and Defect%
\thanks{Work on this paper was partially supported by NSF grant CCF-1408763. See \paperURL\ for the most recent version of this paper.}
}
\author{
	Hsien-Chih Chang
	\qquad\quad
	Jeff Erickson 
\\[1ex]
Department of Computer Science\\
University of Illinois, Urbana-Champaign\\
\href{mailto:hchang17@illinois.edu,jeffe@illinois.edu}{\{hchang17, jeffe\}@illinois.edu}
}

\date{\today}

\maketitle

\begin{bigabstract}
We prove the first nontrivial worst-case lower bounds for two closely related problems. First, $\Omega(n^{3/2})$ degree-1 reductions, series-parallel reductions, and ΔY transformations are required in the worst case to reduce an $n$-vertex plane graph to a single vertex or edge. The lower bound is achieved by any planar graph with treewidth $\Theta(\sqrt{n})$. Second, $\Omega(n^{3/2})$ homotopy moves are required in the worst case to reduce a closed curve in the plane with $n$ self-intersection points to a simple closed curve.  For both problems, the best upper bound known is $O(n^2)$, and the only lower bound previously known was the trivial $Ω(n)$.

The first lower bound follows from the second using medial graph techniques ultimately due to Steinitz, together with more recent arguments of Noble and Welsh [\emph{J.~Graph Theory} 2000]. The lower bound on homotopy moves follows from an observation by Haiyashi \etal\ [\emph{J.~Knot Theory Ramif.} 2012] that the standard projections of certain torus knots have large \emph{defect}, a~topological invariant of generic closed curves introduced by Aicardi and Arnold.  Finally, we prove that every closed curve in the plane with $n$ crossings has defect $O(n^{3/2})$, which implies that better lower bounds for our algorithmic problems will require different techniques.
\end{bigabstract}


\setcounter{page}{0}
\thispagestyle{empty}
\end{titlepage}

\pagestyle{myheadings}
\markboth{Hsien-Chih Chang and Jeff Erickson}{Electrical Reduction, Homotopy Moves, and Defect}

\newpage

\section{Introduction}

\subsection{Electrical Transformations}

Consider the following set of local operations on plane graphs:
\begin{itemize}\itemsep0pt
\item \emph{leaf reduction}: contract the edge incident to a vertex of degree $1$
\item \emph{loop reduction}: delete the edge incident to a face of degree $1$
\item \emph{series reduction}: contract either edge incident to a vertex of degree $2$
\item \emph{parallel reduction}: delete either edge incident to a face of degree $2$
\item \emph{$\arc{Y}{Δ}$ transformation}: delete a vertex of degree $3$ and connect its neighbors with three new edges
\item \emph{$\arc{Δ}{Y}$ transformation}: delete the edges bounding a face of degree $3$ and join the vertices of that face to a new vertex
\end{itemize}
These six operations consist of three dual pairs, as shown in Figure \ref{F:elec-dual}; for example, any series reduction in a plane graph $G$ is equivalent to a parallel reduction in the dual graph $G^*$.  We refer to leaf reductions and loop reductions as \emph{degree-1} reductions, series reductions and parallel reductions as \emph{series-parallel} reductions, and $\arc{Y}{Δ}$ and $\arc{Δ}{Y}$ transformations as \emph{ΔY transformations}. Following Colin de Verdière \etal~\cite{cgv-rep-96}, we collectively refer to these operations and their inverses as \EMPH{electrical transformations}.

\begin{figure}[htb]
\centering
\includegraphics[scale=0.275]{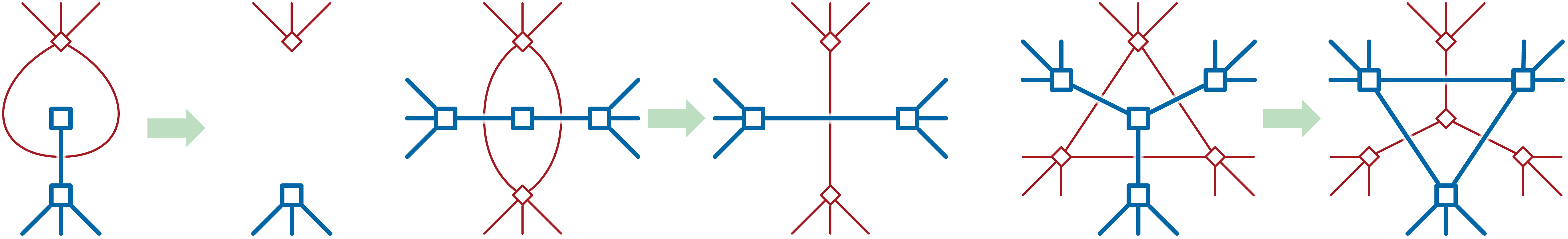}
\caption{Electrical transformations in a plane graph $G$ and its dual graph $G^*$.}
\label{F:elec-dual}
\end{figure}

Electrical transformations have been used since the end of the 19th century~\cite{k-etscn-1899,r-md-1904} to analyze resistor networks and other electrical circuits. Akers~\cite{a-wtns-60} used the same transformations to compute shortest paths and maximum flows (but see also Hobbs~\cite{h-rns-67}); Lehman~\cite{l-wtpn-63} used them to estimate network reliability.  Akers and Lehman both conjectured that any planar graph, two of whose vertices are marked as \EMPH{terminals}, can be reduced to a single edge between the terminals using a finite number of electrical transformations.  This conjecture was first proved by Epifanov~\cite{e-rpges-66} using a nonconstructive argument; simpler constructive proofs were later given by Feo~\cite{f-erpns-85}, Truemper~\cite{t-drpg-89,t-md-92}, Feo and Provan~\cite{fp-dtert-93}, and Nakahara and Takahashi~\cite{nt-aafts-96}.

For the simpler problem of reducing a planar graph without terminals to a single vertex, a constructive proof is already implicit in Steinitz's 1916 proof that every 3-connected planar graph is the 1-skeleton of a 3-dimensional convex polytope~\cite{s-pr-1916,sr-vtp-34}. Grünbaum~\cite{g-cp-67} describes Steinitz's proof in more detail; indeed, Steinitz's proof is often incorrectly attributed to Grünbaum.

These results were later extended to planar graphs with more than two terminals. Gitler~\cite{g-dtaa-91} and Gitler and Sagols~\cite{gs-tdrpg-11} proved that any three-terminal planar graph can be reduced to a graph with three vertices. Archdeacon \etal~\cite{acgp-frpwg-00} and Demasi and Mohar~\cite{dm-ftpdw-15} characterized the four-terminal planar graphs that can be reduced to just four vertices.  Gitler~\cite{g-dtaa-91,cgv-rep-96} proved that for any integer $k$, any planar graph with $k$ terminals on a common face can be reduced to a planar graph with $O(k^2)$ vertices. Gitler's results were significantly extended by Colin de Verdière \etal~\cite{c-rep-92,c-rep1-94,cgv-rep-96} and Curtis \etal~\cite{cmm-fccnb-94,cim-cpgrn-98,cm-ipen-00} to the theory of circular planar networks; see also Kenyon~\cite{k-lpggs-11}.  Similar results have also been proved for several classes of non-planar graphs~\cite{g-dtaa-91,w-drag-15,y-mfmw-06,y-fmwr-04} and matroids~\cite{t-dtm6a-92,w-drag-15}.
%
%


Algorithms for reducing planar graphs using electrical transformations have been applied to several combinatorial problems, including estimating network reliability~\cite{cfp-daptr-96,t-sdtnr-93,sa-refnu-86,gs-dtare-78,t-ndrpg-02};
multicommodity flows~\cite{f-erpns-85}; 
flow estimation from noisy measurements~\cite{zg-effn-07};
counting spanning trees, perfect matchings, and cuts~\cite{cpv-nastc-95,cs-cr-04};
evaluation of spin models in statistical mechanics~\cite{cpv-nastc-95,j-smtra-95};
kinematic analysis of robot manipulators~\cite{st-afkrm-02};
and
solving generalized Laplacian linear systems~\cite{g-cpssd-96,nt-aafts-96}.

%

In light of these numerous applications, it is natural to ask \emph{how many} electrical transformations are required in the worst case to reduce an arbitrary planar graph to a single vertex or edge.  Steinitz's proof~\cite{s-pr-1916,sr-vtp-34,g-cp-67} implies an upper bound of $O(n^2)$, which is the best bound known solely in terms of~$n$.  Feo~\cite{f-erpns-85} and Feo and Provan~\cite{fp-dtert-93} describe reduction algorithms for two-terminal planar graphs that use $O(n^2)$ moves.  In fact, Feo and Provan's algorithm requires at most $O(nD)$ moves, where~$D$ is the diameter of the vertex-face incidence graph (otherwise known as the \emph{radial graph}) of the input graph; $D = Ω(n)$ in the worst case. Feo and Provan~\cite{fp-dtert-93} suggested that “there are compelling reasons to think that $O(\abs{V}^{3/2})$ is the smallest possible order”, possibly referring to earlier empirical results of Feo~\cite[Chapter 6]{f-erpns-85}.  Gitler~\cite{g-dtaa-91} conjectured that a simple modification of Feo and Provan's algorithm requires only $O(n^{3/2})$ time.  Finally, Song~\cite{s-iifpd-01} observed that a naïve implementation of Feo and Provan's algorithm can actually require $\Omega(n^2)$ time, even for graphs that can be reduced using only $O(n)$ steps.

Even the special case of regular grids is open and interesting. Truemper~\cite{t-drpg-89,t-md-92} describes a method to reduce the $p\times p$ grid, or any minor thereof, in $O(p^3)$ steps. Nakahara and Takahashi~\cite{nt-aafts-96} prove an upper bound of $O(\min\set{pq^2, p^2q})$ for any minor of the $p\times q$ cylindrical grid. Since every $n$-vertex planar graph is a minor of an $O(n)\times O(n)$ grid~\cite{v-ucvc-81,s-mncpe-84}, both of these results imply an $O(n^3)$-time algorithm for arbitrary planar graphs; Feo and Provan~\cite{fp-dtert-93} claim without proof that Truemper's algorithm actually performs only $O(n^2)$ electrical transformations.  On the other hand, the smallest (cylindrical) grid containing every $n$-vertex planar graph as a minor has size $Ω(n) \times Ω(n)$~\cite{v-ucvc-81}.  Archdeacon \etal~\cite{acgp-frpwg-00} asked whether the upper bound for square grids can be improved:
\begin{quote}\small
It is possible that a careful implementation and analysis of the grid-embedding schemes can lead to an $O(n\sqrt{n})$-time algorithm for the general planar case. It would be interesting to obtain a near-linear algorithm for the grid\dots. However, it may well be that reducing planar grids is $Ω(n\sqrt{n})$.
\end{quote}

\subsection{Homotopy Moves}

Now consider instead the following set of local operations on closed curves in the plane:
\begin{itemize}\itemsep0pt
\item \EMPH{$\arc{1}{0}$:} Remove an empty loop
\item \EMPH{$\arc{2}{0}$:} Separate two subpaths that bound an empty bigon 
\item \EMPH{$\arc{3}{3}$:} Flip an empty triangle by moving one subpath over the opposite intersection point
\end{itemize}
Our notation is nonstandard but mnemonic; the numbers before and after each arrow indicate the number of local vertices before and after the move.  See Figure \ref{F:homotopy}.  Each of these operations can be performed by continuously deforming the curve within a small neighborhood of one face; consequently, we call these transformations and their inverses \EMPH{homotopy moves}.  Homotopy moves are “shadows” of the classical Reidemeister moves used to manipulate knot and link diagrams~\cite{ab-tkc-26,r-ebk-27}.  A compactness argument, first explicitly given by Titus~\cite{t-ctafb-61} and Francis~\cite{f-thnc-71,f-ghi-71} but implicit in earlier work of Alexander and Briggs~\cite{ab-tkc-26} and Reidemeister~\cite{r-ebk-27}, implies that any continuous deformation between two generic closed curves in the plane—or on any other surface—is equivalent to a finite sequence of homotopy moves.   

\begin{figure}[htb]
\centering
\includegraphics[scale=0.275]{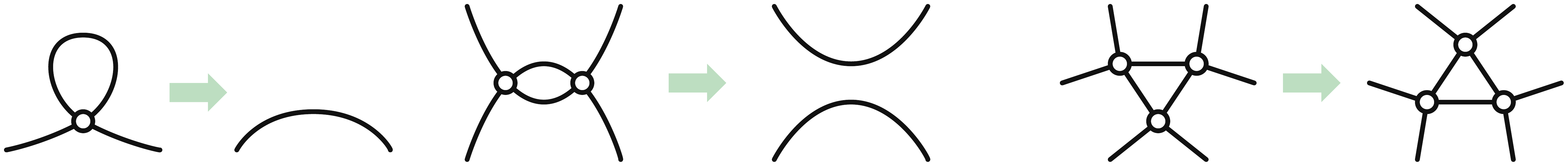}\\
\caption{Homotopy moves $\arc10$, $\arc20$, and $\arc33$.}
\label{F:homotopy}
\end{figure}

It is natural to ask \emph{how many} homotopy moves are required to transform a given closed curve in the plane into a \emph{simple} closed curve.  An $O(n^2)$ upper bound follows by suitable modifications of the Steinitz-Grünbaum and Feo-Provan algorithms for electrical reduction, where $n$ is the number of self-intersection points of the given curve.  (See Lemma~\ref{L:homotopy}.)  The same $O(n^2)$ bound also follows from algorithms for \emph{regular} homotopy, which allows only $\biarc20$ and $\arc33$ homotopy moves, by Francis~\cite{f-frtcs-69}, Vegter~\cite{v-kfdp-89} (for polygonal curves), and Nowik~\cite{n-cpsc-09}.

Tight bounds are known for two restrictions of this question, where some types of homotopy moves are forbidden.  First, Nowik~\cite{n-cpsc-09} proved that $Ω(n^2)$ moves are necessary in the worst case to connect two regularly homotopic curves with $n$ vertices.  Second, Khovanov~\cite{k-dg-97} defines two curves to be \emph{doodle equivalent} if one can be transformed into the other using $\biarc10$ and $\biarc20$ homotopy moves.  Khovanov~\cite{k-dg-97} and Ito and Takimura~\cite{it-whkp-13} independently proved that each doodle-equivalence class contains a unique representative with the smallest number of vertices, and that any curve can be transformed into the simplest doodle equivalent curve using only $\arc10$ and $\arc20$ moves.  It follows that two doodle equivalent curves are connected by a sequence of only $O(n)$ homotopy moves.\footnote{It is not known which sets of curves are equivalent under $\biarc10$ and $\arc33$ moves; indeed, Hagge and Yazinski only recently proved that this equivalence is nontrivial~\cite{hy-nrmsi-14}; see also related results of Ito \etal~\cite{it-whkp-13,itt-swhkp-15}.}

Looser bounds are also known for the minimum number of Reidemeister moves needed to reduce a diagram of the unknot~\cite{hn-udrqn-10,l-pubrm-15}, to separate the components of a split link~\cite{hhn-unnrm-12}, or to move between two equivalent knot diagrams~\cite{hh-msrmd-11,cl-ubrm-14}.

\subsection{Our Results}

In this paper, we prove the first non-trivial lower bounds for both of these problems.  Specifically:
\begin{itemize}
\item
$\Omega(n^{3/2})$ electrical transformations are required in the worst case to reduce a plane graph with $n$ vertices, with or without terminals, to a single vertex (or any constant number of vertices).
\item
$\Omega(n^{3/2})$ homotopy moves are required in the worst case to reduce a generic closed curve in the plane with $n$ self-intersection points to a simple closed curve.
\end{itemize}

Like many other authors, starting with Steinitz~\cite{s-pr-1916,sr-vtp-34} and Grünbaum\cite{g-cp-67}, we study electrical transformations indirectly, through the lens of \emph{medial graphs}.  By refining arguments of Noble and Welsh~\cite{nw-kg-00} and others, we prove in Section \ref{S:reduction} that the minimum number of electrical transformations needed to completely reduce a plane graph $G$ to a single vertex is no smaller than the minimum number of homotopy moves required to transform its medial graph into a simple closed curve.  Thus, our first lower bound follows immediately from the second.

Our lower bound for homotopy moves relies on a topological invariant called \EMPH{defect}, which was introduced by Arnold~\cite{a-tipcc-94, a-pctip-94} and~Aicardi~\cite{a-tc-94}.  Every simple closed curve has defect zero, and any homotopy move changes the defect of a curve by $-2$, $0$, or $2$; the various cases are illustrated in Figure~\ref{F:defect}.  In Section \ref{S:lower-bound}, we compute the defect of the standard planar projection of any $p\times q$ torus knot where either $p\bmod q = 1$ or $q\bmod p = 1$, generalizing earlier results of Hayashi \etal~\cite{hh-msrmd-11,hhsy-musrm-12} and Even-Zohar \etal~\cite{ehln-irkl-14}.  In particular, we show that the standard projection of the $p\times (p+1)$ torus knot, which has $p^2-1$ vertices, has defect $2\binom{p+1}{3}$.

Putting all the pieces together, we conclude that for any integer $k$, reducing the $k \times (2k+1)$ cylindrical grid requires at least $\binom{2k+1}{3} \ge (\sqrt{2}/{3})n^{3/2} - O(n)$ electrical transformations.  An argument of Truemper~\cite[Lemma~4]{t-drpg-89} implies that if $H$ is any minor of a planar graph $G$, then $H$ requires no more electrical transformations to reduce than $G$; see Lemma \ref{L:smoothing}.  It follows that our lower bound applies to any planar graph with treewidth $Ω(\sqrt{n})$~\cite{rst-qepg-94}; in particular, Truemper's $O(p^3)$ bound for the $p\times p$ grid~\cite{t-drpg-89,t-md-92} is tight.  Our analysis also implies that for any integers $p$ and $q$, electrically reducing the $p\times q$ cylindrical grid requires $\Omega(\min\set{p^2q,pq^2})$ moves, matching Nakahara and Takahashi’s upper bound~\cite{nt-aafts-96}.

Finally, in Section \ref{S:upper-bound}, we prove that the defect of any generic closed curve $γ$ with $n$ vertices has absolute value at most $O(n^{3/2})$. Unlike most $O(n^{3/2})$ upper bounds involving planar graphs, our proof does \emph{not} use the planar separator theorem~\cite{lt-stpg-79}.  Feo and Provan's electrical reduction algorithm~\cite{fp-dtert-93} and our electrical-to-homotopy reduction (Section \ref{S:reduction}) imply an upper bound of $O(nD)$, where $D$ is the diameter of the dual graph of $γ$; we give a simpler self-contained proof of this upper bound in Section~\ref{SS:diameter}.  Thus, if $D=O(\sqrt{n})$, we are done.  Otherwise, we prove that there is a simple closed curve $σ$ with at least~$s^2$ vertices of $γ$ on either side, where $s$ is the number of times $σ$ crosses $γ$. In Section \ref{SS:inex} we establish an inclusion-exclusion relationship between the defects of the given curve $γ$, the curves obtained by simplifying $γ$ either inside or outside $σ$, and the curve obtained by simplifying $γ$ on both sides of~$σ$.  This relationship implies an unbalanced “divide-and-conquer” recurrence whose solution is $O(n^{3/2})$.

Our upper bound on defect implies that better worst-case lower bounds on the number of electrical transformations or homotopy moves would require new techniques.  Like Gitler~\cite{g-dtaa-91}, Feo and Provan~\cite{fp-dtert-93}, and Archdeacon \etal~\cite{acgp-frpwg-00}, we conjecture that the correct worst-case bound for both problems is $Θ(n^{3/2})$. 

%

\subsection{Saving Face}

Electrical transformations are usually defined more generally, without references to a planar embedding of the underlying graph. In this more general setting, a loop reduction deletes any loop (even if it does not bound a face), a parallel reduction deletes any edge parallel to another edge (even if those two edges do not bound a face), and a $\arc{Δ}{Y}$ transformation removes the edges of any 3-cycle (even if it does not bound a face) and connects its vertices to a new vertex.  Our lower bound technique requires our more restrictive definitions; on the other hand, all published algorithms for reducing planar graphs also require only electrical transformations meeting our definition. 

Our argument can be extended to allow non-facial loop reductions (or, by duality, contraction of any bridge) and non-facial parallel reductions (or, by duality, contraction of either edge in an edge cut of size~$2$). However, non-facial $\arc{Δ}{Y}$ transformations are more problematic, because they can destroy the planarity of the graph.  For example, a single $\arc{Δ}{Y}$ transformation transforms the planar graph obtained from $K_5$ by deleting one edge into the non-planar graph $K_{3,3}$.  It is an interesting open problem whether our $Ω(n^{3/2})$ lower bound holds even when non-planar $\arc{Δ}{Y}$ transformations are permitted.

\section{Definitions}

\subsection{Closed Curves}

A \EMPH{(generic) closed curve} is a continuous map $γ\colon S^1 \to \Real^2$ that is injective except at a finite number of self-intersections, each of which is a transverse double point. More concisely, we consider only \emph{generic immersions} of the circle into the plane. A closed curve is \EMPH{simple} if it is injective.

The image of any non-simple closed curve has a natural structure as a 4-regular plane graph.   Thus, we refer to the self-intersection points of a curve as its \EMPH{vertices}, the maximal subpaths between vertices as \EMPH{edges}, and the components of the complement of the curve as its \EMPH{faces}.  In particular, to avoid trivial boundary cases, we consider a \emph{simple} closed curve to be a single edge with no vertices.  Conversely, every 4-regular planar graph is the image of a generic immersion of one or more disjoint circles.  We call a 4-regular plane graph \EMPH{unicursal} if it is the image of a generic closed curve.

Intuitively, the \EMPH{winding number} of a closed curve $γ$ around a point $x$, which we denote \EMPH{$\Wind(γ,x)$}, is the number of times $γ$ travels counterclockwise around $x$.  Winding numbers are characterized by the following combinatorial rule, first proposed by Möbius~\cite{m-ubiep-1865} but widely known as \emph{Alexander numbering}~\cite{a-tikl-28}: The winding numbers of $γ$ around any two points in the same face of $γ$ are equal; the winding number around any point in the outer face is zero; and winding numbers around points in adjacent faces differ by~$1$, with the larger winding number appearing on the left side of the curve.  See Figure~\ref{F:Alexander} for an example.  By convention, the winding number around any point $x$ in the image of $γ$ is the average of the winding numbers around all faces incident to $x$.  The winding number around any vertex is still an integer; the winding number around any regular point of $γ$ is a half-integer.

\begin{figure}[htb]
\centering
\includegraphics[scale=0.4]{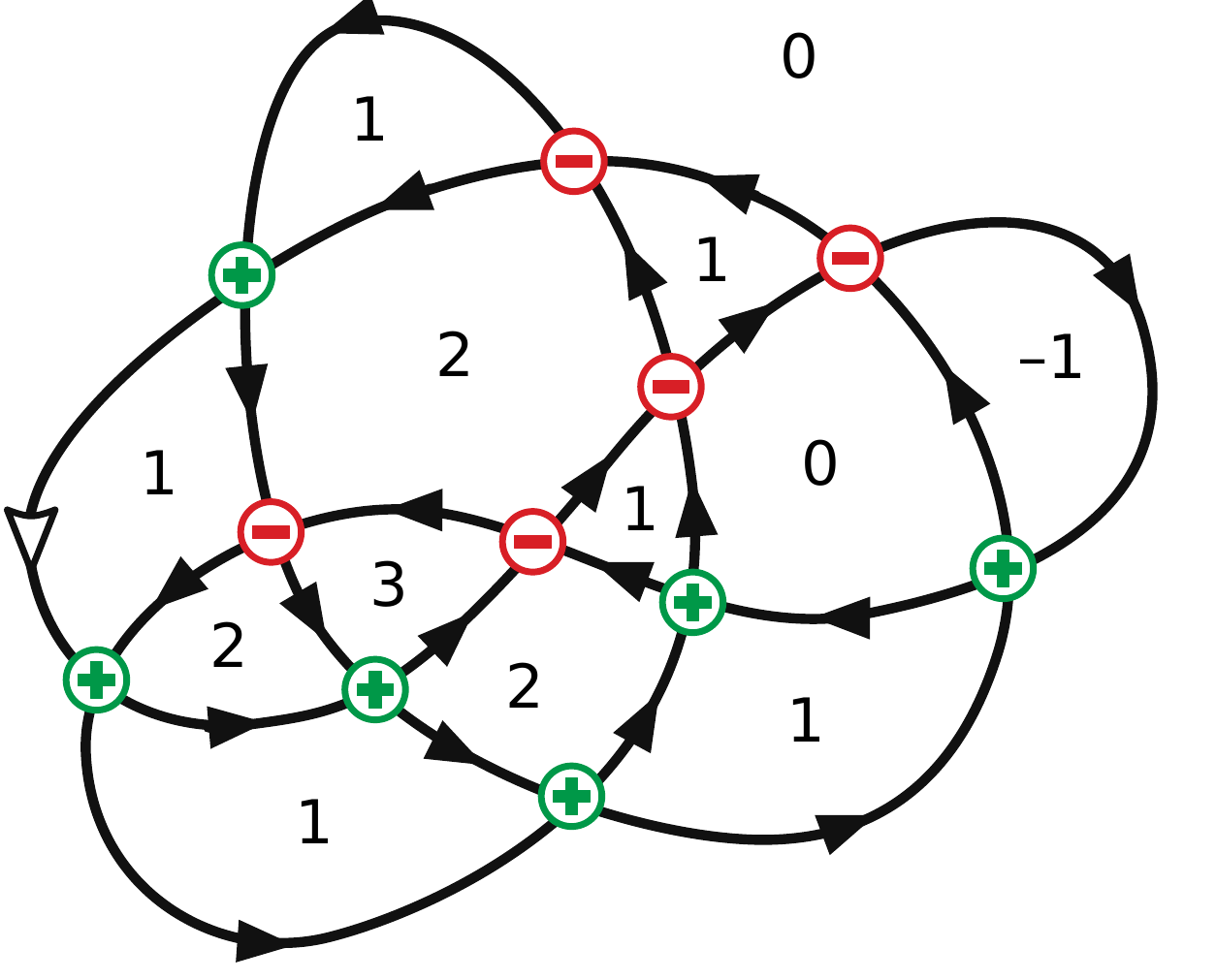}
\caption{%
Alexander numbering and vertex signing for a curve.  The white triangle on the left is the basepoint.}
\label{F:Alexander}
\end{figure}

We adopt a standard sign convention for vertices originally proposed by Gauss~\cite{g-n1gs-00}.\footnote{Some authors use the opposite sign convention, both for vertices and for winding number.}  A vertex is \EMPH{positive} if the first traversal through the vertex crosses the second traversal from right to left, and \EMPH{negative} otherwise.  Equivalently, a vertex $x$ is positive if the winding number of a point moving along the curve increases the first time it reaches $x$ after leaving the basepoint. We define $\sgn(x) = +1$ for every positive vertex $x$ and $\sgn(x) = -1$ for every negative vertex~$x$.  Again, see Figure~\ref{F:Alexander}.

\subsection{Medial Graphs}

The \EMPH{medial graph} of a plane graph $G$, which we denote \EMPH{$G^\times$}, is another plane graph whose vertices correspond to the edges of $G$ and whose edges correspond to incidences between vertices of $G$ and faces of $G$. Two vertices of $G^\times$ are connected by an edge if and only if the corresponding edges in $G$ are consecutive in cyclic order around some vertex, or equivalently, around some face in $G$. 
The medial graph~$G^\times$ may contain loops and parallel edges even if the original graph $G$ is simple.  The medial graphs of any plane graph~$G$ and its dual $G^*$ are identical.  Every vertex in every medial graph has degree $4$, and every 4-regular plane graph is a medial graph.  To avoid trivial boundary cases, we define the medial graph of an isolated vertex to be a circle, which we regard as an edge with no vertices.

\begin{figure}[htb]
\centering
\includegraphics[scale=0.275]{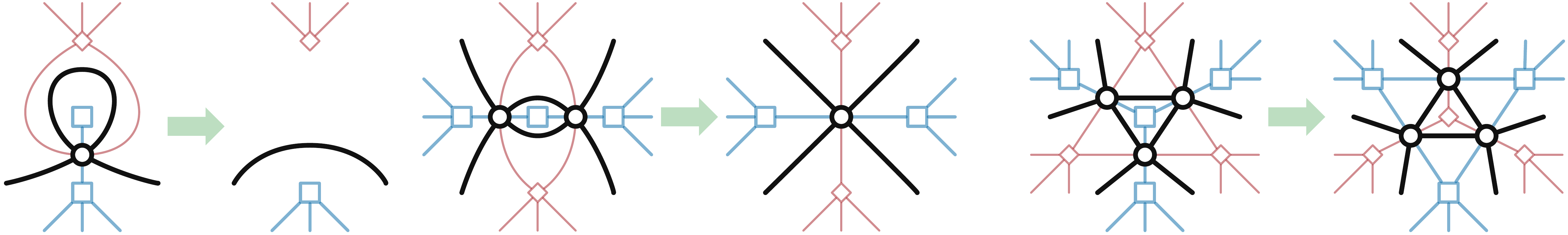}\\
\caption{Medial electrical moves $\arc10$, $\arc21$, and $\arc33$.}
\end{figure}

Electrical transformations in any planar graph $G$ correspond to local transformations in the medial graph $G^\times$, which are almost identical to homotopy moves.  Each leaf or loop reduction in $G$ corresponds to to a $\arc10$ homotopy move in $G^\times$, and each $\arc{Δ}{Y}$ or $\arc{Y}{Δ}$ transformation in $G$ corresponds to a $\arc33$ homotopy move in~$G^\times$.  A series-parallel reduction in $G$ contracts an empty bigon in $G^\times$ to a single vertex.  Extending our earlier notation, we call this transformation a \EMPH{$\arc21$} move.  We collectively refer to these transformations and their inverses as \EMPH{medial electrical moves}.


%

\section{Electrical Reduction is No Shorter than Homotopy Reduction}
\label{S:reduction}

We say that a sequence of elementary moves (of either type) \EMPH{reduces} a 4-regular plane graph~$γ$ if it transforms~$γ$ into a collection of disjoint simple closed curves.  We define two functions describing the minimum number of moves required to reduce~$γ$:
\begin{itemize}\itemsep0pt
\item 
\EMPH{$X(γ)$} is the minimum number of medial electrical moves required to reduce $γ$.
\item
\EMPH{$H(γ)$} is the minimum number of homotopy moves required to reduce $γ$.
\end{itemize}
The main result of this section is that the first function is always an upper bound on the second when $γ$ is a generic closed curve. This result is already implicit in the work of Noble and Welsh~\cite{nw-kg-00}, and most of our proofs closely follow theirs. We include the proofs here to make the inequalities explicit and to keep the paper self-contained.

\EMPH{Smoothing} a 4-regular plane graph $γ$ at a vertex $x$ means replacing the intersection of $γ$ with a small neighborhood of $x$ with two disjoint simple paths, so that the result is another 4-regular plane graph. (There are two possible smoothings at each vertex; see Figure \ref{F:smoothing}. A \EMPH{smoothing} of $γ$ is any graph obtained by smoothing zero or more vertices of $γ$, and a \EMPH{proper smoothing} of $γ$ is any smoothing other than $\gamma$ itself. For any plane graph $G$, the (proper) smoothings of the medial graph $G^\times$ are precisely the medial graphs of (proper) minors of $G$.

\begin{figure}[htb]
\centering
\includegraphics[scale=0.25]{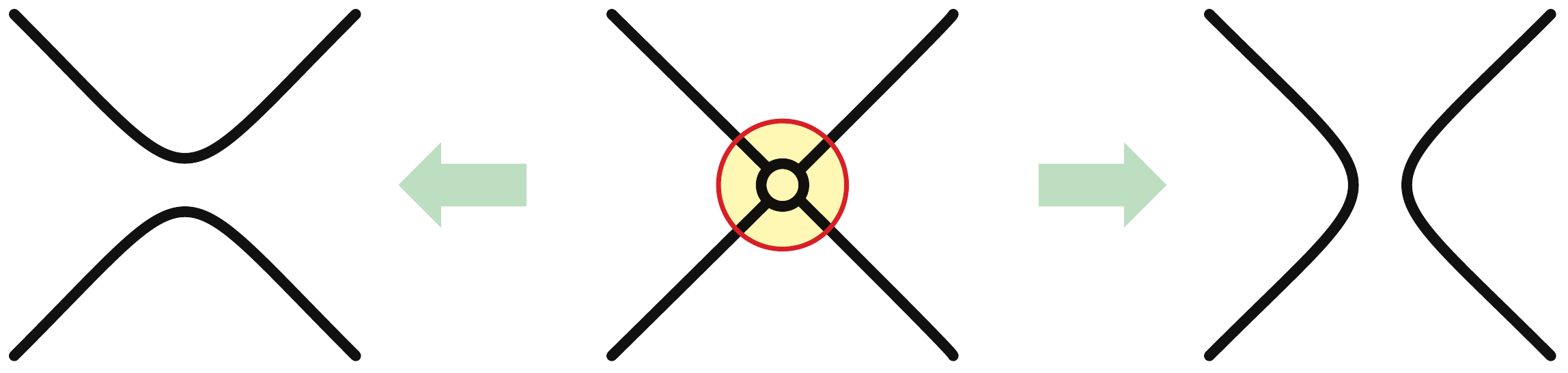}
\caption{Smoothing a vertex.}
\label{F:smoothing}
\end{figure}

The next lemma follows from close reading of proofs by Truemper~\cite[Lemma~4]{t-drpg-89} and several others~\cite{g-dtaa-91,nt-aafts-96,acgp-frpwg-00,nw-kg-00} that every minor of a ΔY-reducible graph is also ΔY-reducible.  Our proof most closely resembles an argument of Gitler~\cite[Lemma~2.3.3]{g-dtaa-91}, but restated in terms of medial electrical moves.


\begin{lemma}
\label{L:smoothing}
$X(\overline{γ}) < X(γ)$ for every connected proper smoothing $\overline{γ}$ of every connected 4-regular plane graph $γ$.
\end{lemma}

\begin{proof}
Let $γ$ be a connected 4-regular plane graph, and let $\overline{γ}$ be a connected proper smoothing of $γ$.  If $γ$ is already simple, the lemma is vacuously true.  Otherwise, the proof proceeds by induction on $X(γ)$.

We first consider the special case where $\overline{γ}$ is obtained from $γ$ by smoothing a single vertex~$x$.  Let~$γ'$ be the result of the first medial electrical move in the minimum-length sequence that reduces $γ$.  We immediately have $X(γ) = Χ(γ')+1$.  There are two nontrivial cases to consider.

First, suppose the move from $γ$ to $γ'$ does not involve the smoothed vertex $x$.  Then we can apply the same move to $\overline{γ}$ to obtain a new graph $\overline{γ}'$; the same graph  can also be obtained from $γ'$ by smoothing~$x$.  We immediately have $X(\overline{γ}) \le X(\overline{γ}') + 1$, and the inductive hypothesis implies $X(\overline{γ}') < X(γ')$.

Now suppose the first move in $Σ$ does involve $x$. In this case, we can apply at most one medial electrical move to $\overline{γ}$ to obtain a (possibly trivial) smoothing $\overline{γ}'$ of $γ'$.  There are eight subcases to consider, shown in Figure \ref{F:smooth-moves}.  One subcase for the $\arc01$ move is impossible, because~$\overline{γ}$ is connected.  In the remaining $\arc01$ subcase and one $\arc21$ subcase, the curves $\overline{γ}$, $\overline{γ}'$ and $γ'$ are all isomorphic, which implies $X(\overline{γ}) = X(\overline{γ}') = X(γ') = X(γ)-1$.  In all remaining subcases, $\overline{γ}'$ is a connected proper smoothing of $γ'$, so the inductive hypothesis implies $X(\overline{γ}) ≤ X(\overline{γ}')+1 < X(γ')+1 = X(γ)$.

\begin{figure}[htb]
\centering
\includegraphics[scale=0.25]{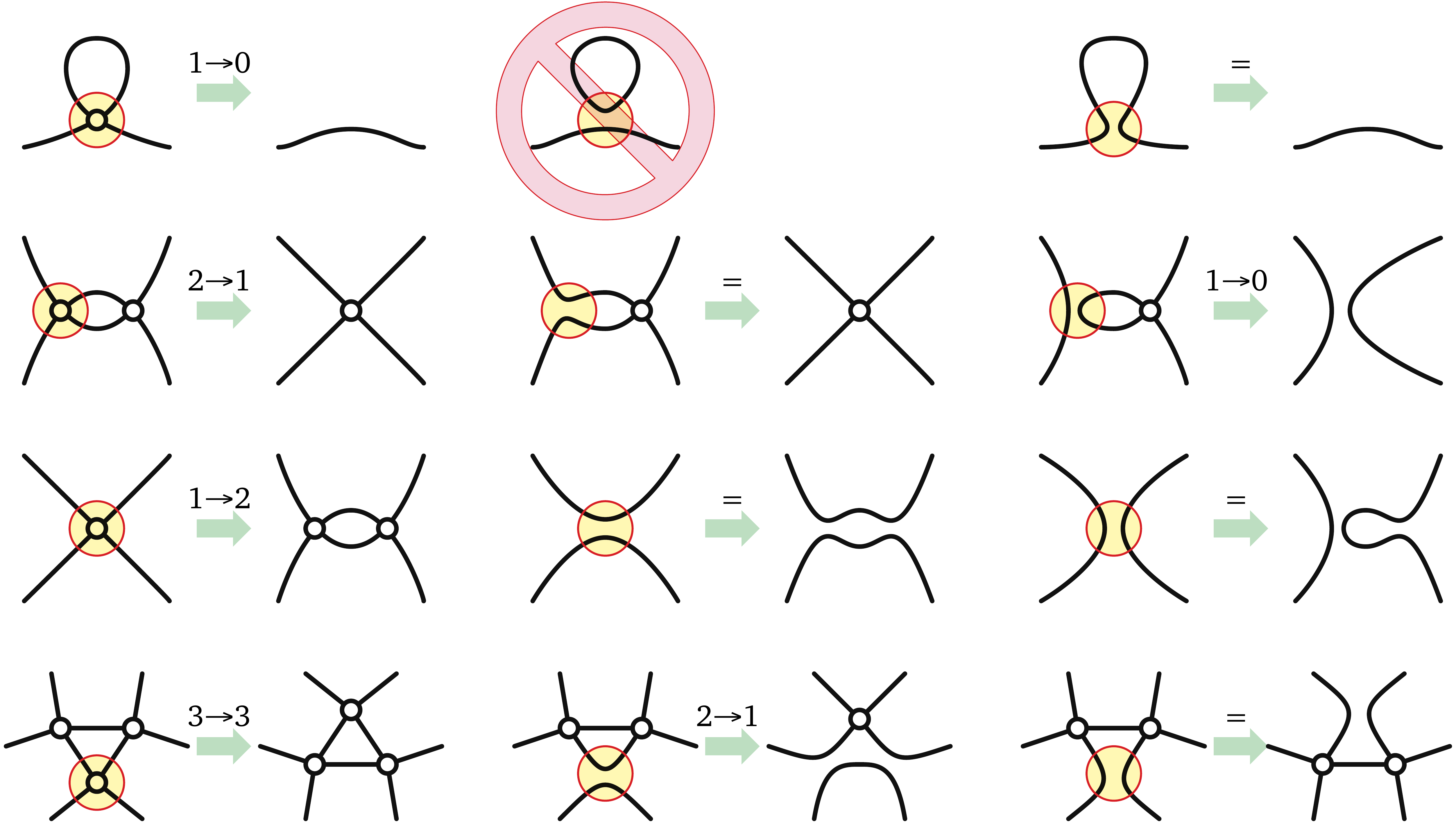}
\caption{Cases for the proof of the Lemma~\ref{L:smoothing}; the circled vertex is $x$.}
\label{F:smooth-moves}
\end{figure}

Finally, we consider the more general case where $\overline{γ}$ is obtained from $γ$ by smoothing more than one vertex.  Let $\widetilde{γ}$ be any intermediate curve, obtained from $γ$ by smoothing just one of the vertices that were smoothed to obtain $\overline{γ}$.  Our earlier argument implies that $X(\widetilde{γ}) < X(γ)$.  Thus, the inductive hypothesis implies $X(\overline{γ}) < X(\widetilde{γ})$, which completes the proof.
\end{proof}

\begin{lemma}
\label{L:monotonicity}
For every connected 4-regular plane graph $γ$, there is a minimum-length sequence of medial electrical moves that reduces $γ$ and that does not contain $\arc01$ or $\arc12$ moves. 
\end{lemma}

\begin{proof}
Our proof follows an argument of Noble and Welsh~\cite[Lemma~3.2]{nw-kg-00}.

Consider a minimum-length sequence of medial electrical moves that reduces an arbitrary connected $4$-regular planar graph~$γ$. For any integer $i ≥ 0$, let $γ_i$ denote the result of the first $i$ moves in this sequence; in particular, $γ_0 = γ$ and $γ_{X(γ)}$ is a set of disjoint circles. Minimality of the reduction sequence implies that $X(γ_i) = X(γ)-i$ for all $i$. Now let $i$ be an arbitrary index such that $γ_i$ has one more vertex than $γ_{i-1}$. Then $γ_{i-1}$ is a connected proper smoothing of $γ_i$, so Lemma~\ref{L:smoothing} implies that $X(γ_{i-1}) < X(γ_i)$, giving us a contradiction. 
\end{proof}

\begin{theorem}
\label{Th:homotopy}
$X(γ) ≥ H(γ)$ for every closed curve $γ$.
\end{theorem}

\begin{proof}
The proof proceeds by induction on $X(γ)$, following an argument of Noble and Welsh~\cite[Proposition 3.3]{nw-kg-00}

Let $γ$ be a closed curve.  If $X(γ) = 0$, then $γ$ is already simple, so $H(γ) = 0$.  Otherwise, let $Σ$ be a minimum-length sequence of medial electrical moves that reduces $γ$ to a circle.  Lemma~\ref{L:monotonicity} implies that we can assume that the first move in $Σ$ is neither $\arc01$ nor $\arc12$. If the first move is $1\arcto 0$ or $3\arcto 3$, the theorem immediately follows by induction. 

The only interesting first move is $2\arcto 1$. Let $γ'$ be the result of this $\arc21$ move, and let $\overline{γ}$ be the result of the corresponding $\arc20$ homotopy move.  The minimality of $Σ$ implies that $X(γ) = X(γ')+1$, and we trivially have $H(γ) \le H(\overline{γ}) + 1$. The curve $\overline{γ}$ is a connected proper smoothing of $γ'$, so the Lemma~\ref{L:smoothing} implies $X(\overline{γ}) < X(γ') < X(γ)$.  Finally, the inductive hypothesis implies that $X(\overline{γ}) \ge H(\overline{γ})$, which completes the proof.
\end{proof}


Finally, Theorem \ref{Th:homotopy} and Lemma \ref{L:smoothing} immediately imply the following useful corollary.

\begin{corollary}
\label{C:smoothing}
$X(γ) ≥ H(\overline{γ})$ for every connected 4-regular plane graph $γ$ and every unicursal smoothing~$\overline{γ}$  of $γ$.
\end{corollary}


\section{Lower Bounds}
\label{S:lower-bound}

\subsection{Defect}
\label{SS:defect}

To complete our lower bound proof, we consider an invariant of closed curves in the plane introduced by Arnold~\cite{a-tipcc-94, a-pctip-94} and~Aicardi~\cite{a-tc-94} called \EMPH{defect}. (Some readers may be more comfortable thinking of defect as a \emph{potential function} for 4-regular plane graphs.) There are several equivalent definitions and closed-form formulas for defect and other closely related curve invariants \cite{a-tipcc-94, a-pctip-94, v-gicsc-95, s-efspc-95, lw-igopc-96, cd-efagci-97, l-nicsp-97, ao-gasi-99, p-icfgd-98, p-nwfpc-99}; the most useful formula for our purposes is due to Polyak~\cite{p-icfgd-98}. 



Two vertices $x\ne y$ of a closed curve $γ$ are \EMPH{interleaved} if they alternate in cyclic order along $γ$, either as $x,y,x,y$ or as $y,x,y,x$; we write \EMPH{$x\between y$} to denote that vertices $x$ and $y$ are interleaved.  Polyak's formula for defect is
\[
	δ(γ) \coloneqq -2 \sum_{x\between y} \sgn(x)\cdot\sgn(y),
\]
where the sum is taken over all interleaved pairs of vertices.  The factor of $-2$ is a historical artifact, which we retain only to be consistent with Arnold's original definitions~\cite{a-tipcc-94, a-pctip-94}.  See Figure \ref{F:defect} for an example.  Even though the signs of individual vertices depend on the basepoint and orientation of the curve, the defect of a curve is independent of those choices. Moreover, the defect of any curve is preserved by any homeomorphism from the plane to itself, or even from the sphere to itself, including reflection. 

\def\REdgeG#1{\multicolumn{1}{|c}{\color{Green}#1}}
\def\REdgeR#1{\multicolumn{1}{|c}{\color{Red}#1}}
\def\LEdge#1{\multicolumn{1}{c}{#1}}

\begin{figure}[htb]
\centering\scriptsize
	\raisebox{-0.5\height}{\includegraphics[scale=0.4]{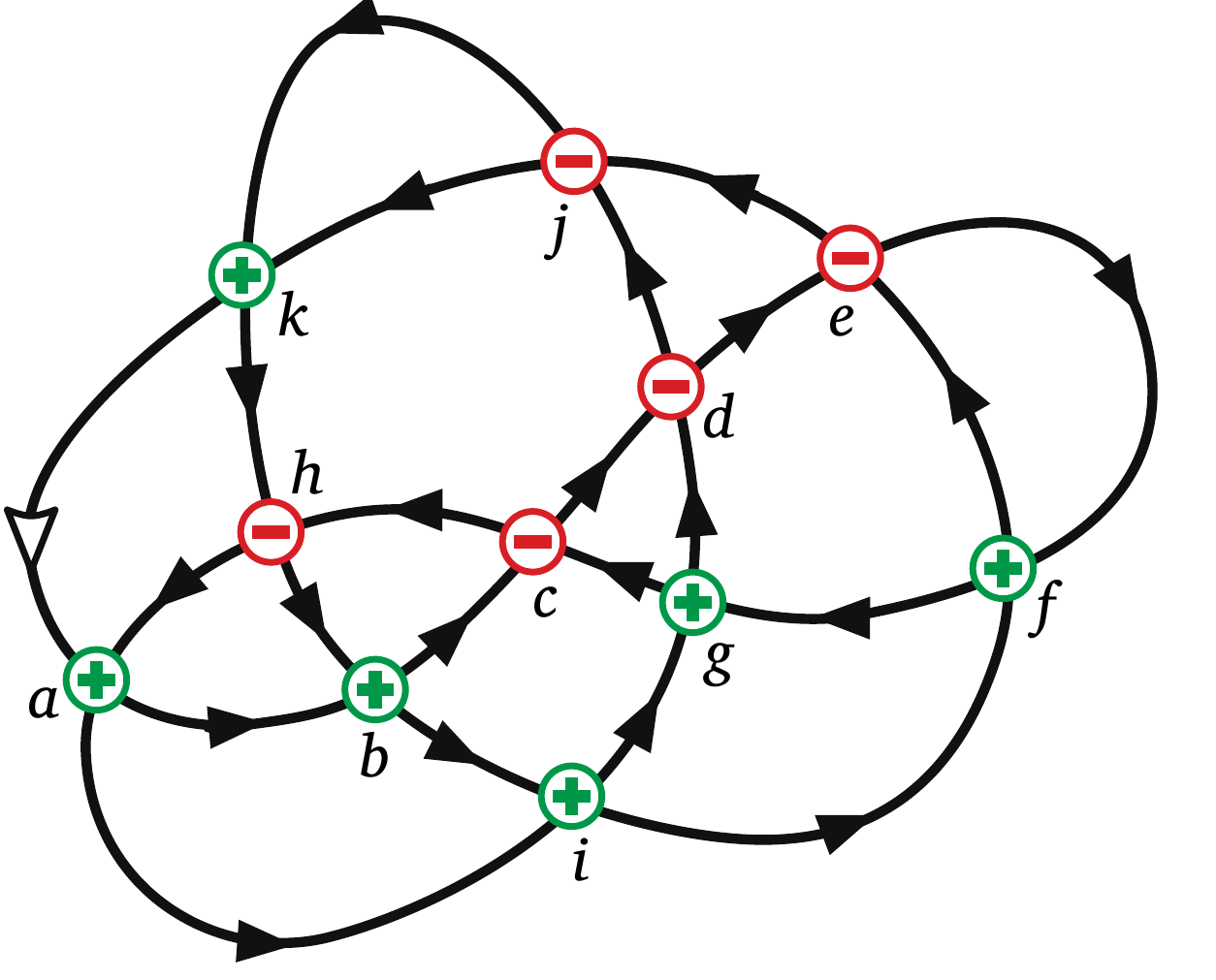}}
\qquad
\arraycolsep=3.5pt\def\arraystretch{1.2}
$\begin{array}{|cccccccccccc}
\\[-6ex]
\LEdge{\color{Green}a} \\ \cline{1-1}
+	& \REdgeG{b} & \\\cline{2-2}
	& 	& \REdgeR{c} \\\cline{3-3}
-	& 	& +	& \REdgeR{d} \\\cline{4-4}
-	& -	& +	& +	& \REdgeR{e} \\\cline{5-5}
+	& +	& -	& -	& 	& \REdgeG{f} \\\cline{6-6}
+	& 	& -	& 	& 	& 	& \REdgeG{g} \\\cline{7-7}
- 	& 	& 	& +	& 	& 	& -	& \REdgeR{h} \\\cline{8-8}
	& +	& 	& -	& 	& 	& +	& -	& \REdgeG{i} \\\cline{9-9}
	& -	& 	& 	& +	& -	& 	& +	& -	& \REdgeR{j} \\\cline{10-10}
	& +	& 	& 	& -	& +	& 	& -	& +	& -	& \REdgeG{k} \\\cline{1-10}
\end{array}$
\caption{The curve in Figure \ref{F:Alexander} has defect $-2(15-16) = 2$.  In the table on the right, each $+$ indicates an interleaved pair with the same sign, each $-$ indicates an interleaved pair with opposite signs, and blanks indicate non-interleaved pairs.}
\label{F:defect}
\end{figure}

Trivially, every simple closed curve has defect zero. Straightforward case analysis~\cite{p-icfgd-98} implies that any single homotopy move changes the defect of a curve by at most $2$:
\begin{itemize}\itemsep0pt
\item A $\arc10$ move leaves the defect unchanged.
\item A $\arc20$ move decreases the defect by $2$ if the two disappearing vertices are interleaved, and leaves the defect unchanged otherwise.
\item A $\arc33$ move increases the defect by $2$ if the three vertices before the move contain an even number of interleaved pairs, and decreases the defect by $2$ otherwise.
\end{itemize}
The various cases are illustrated in Figure \ref{F:defect-moves}. 
Theorem \ref{Th:homotopy} now has the following immediate corollary:

\begin{corollary}
\label{C:lowerbound}
$X(γ) ≥ H(γ) ≥ \Abs{δ(γ)}/2$ for every closed curve $γ$.
\end{corollary}

\begin{figure}
\centering\footnotesize
\def\arraystretch{1.25}
\begin{tabular}{c|c:cc:cc}
	Move
	& $1\arcto 0$ & \multicolumn{2}{c:}{$2\arcto 0$} & \multicolumn{2}{c}{$3\arcto 3$}
	\\ \hline
	$\begin{matrix} γ \\ \raisebox{-.5\height}{\rotatebox{90}{$\Longleftarrow$}} \\ γ' \end{matrix}$
	&
	\raisebox{-.5\height}{\includegraphics[scale=0.25]{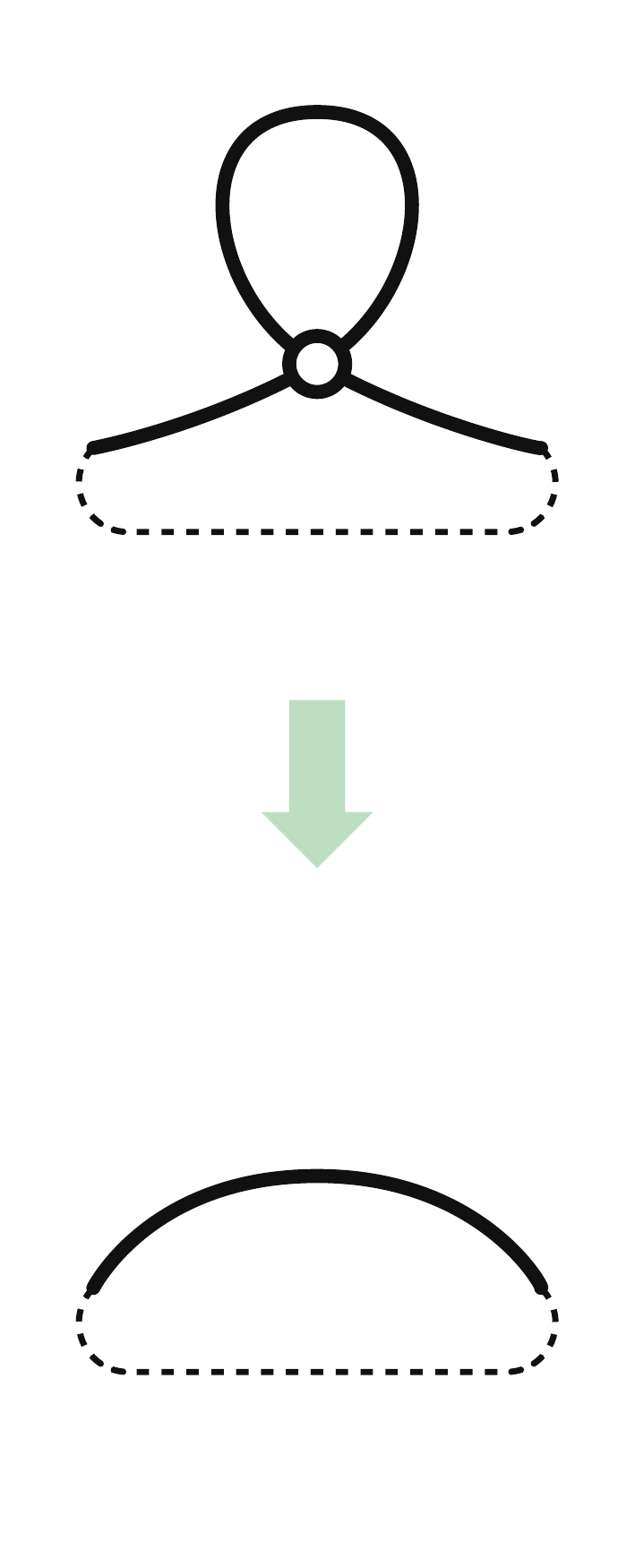}}
	&
	\raisebox{-.5\height}{\includegraphics[scale=0.25]{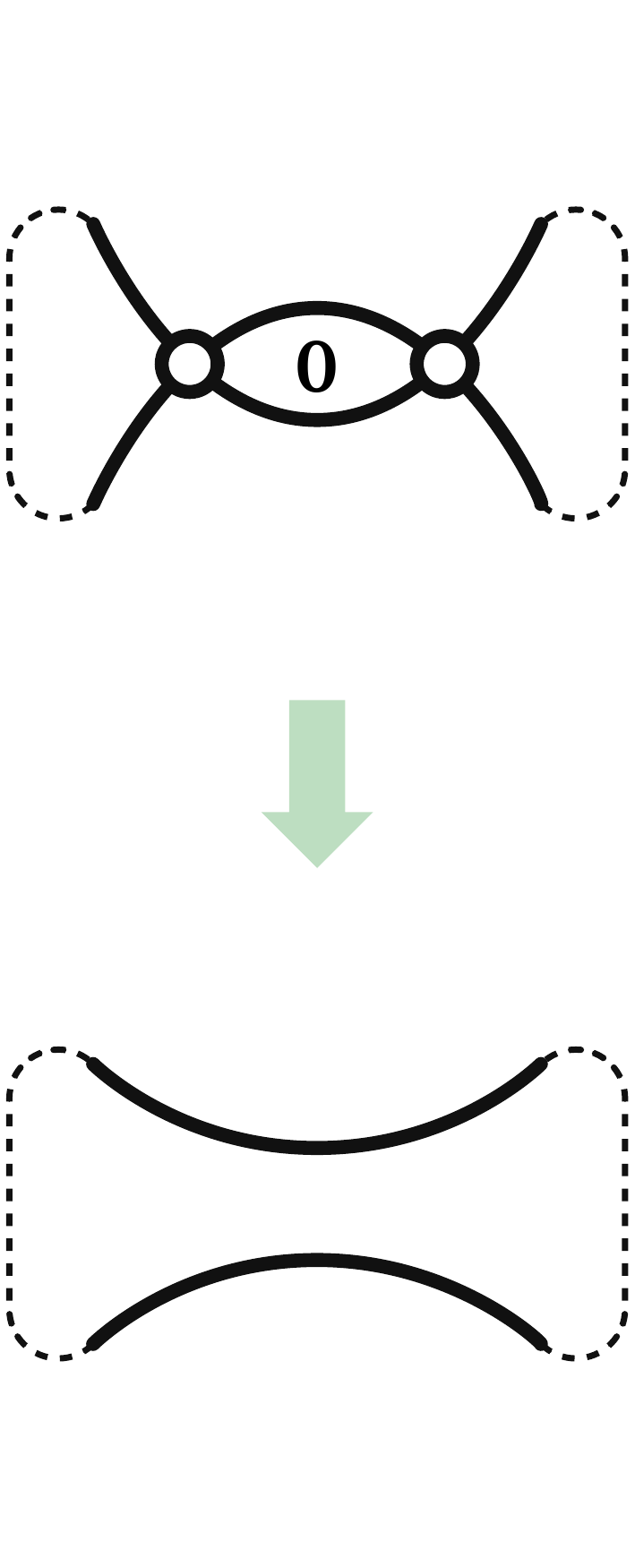}}
	&
	\raisebox{-.5\height}{\includegraphics[scale=0.25]{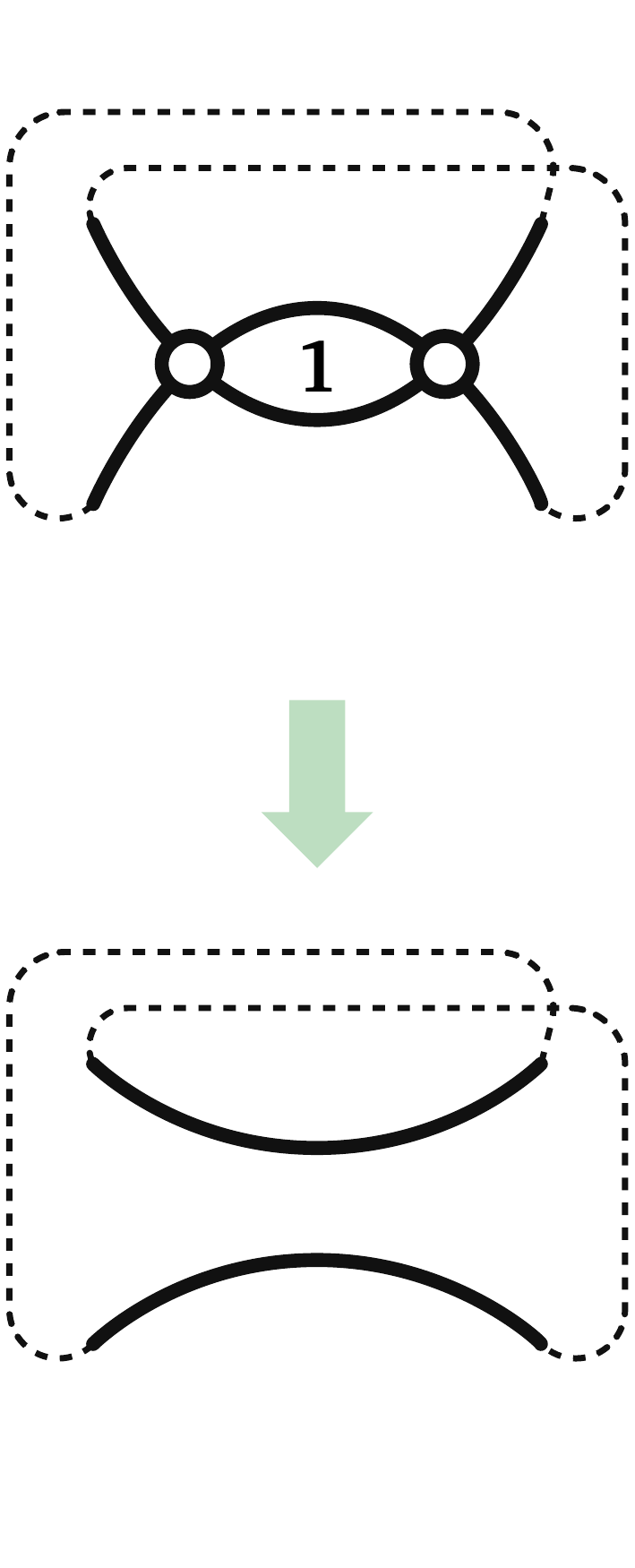}}
	&
	\raisebox{-.5\height}{\includegraphics[scale=0.25]{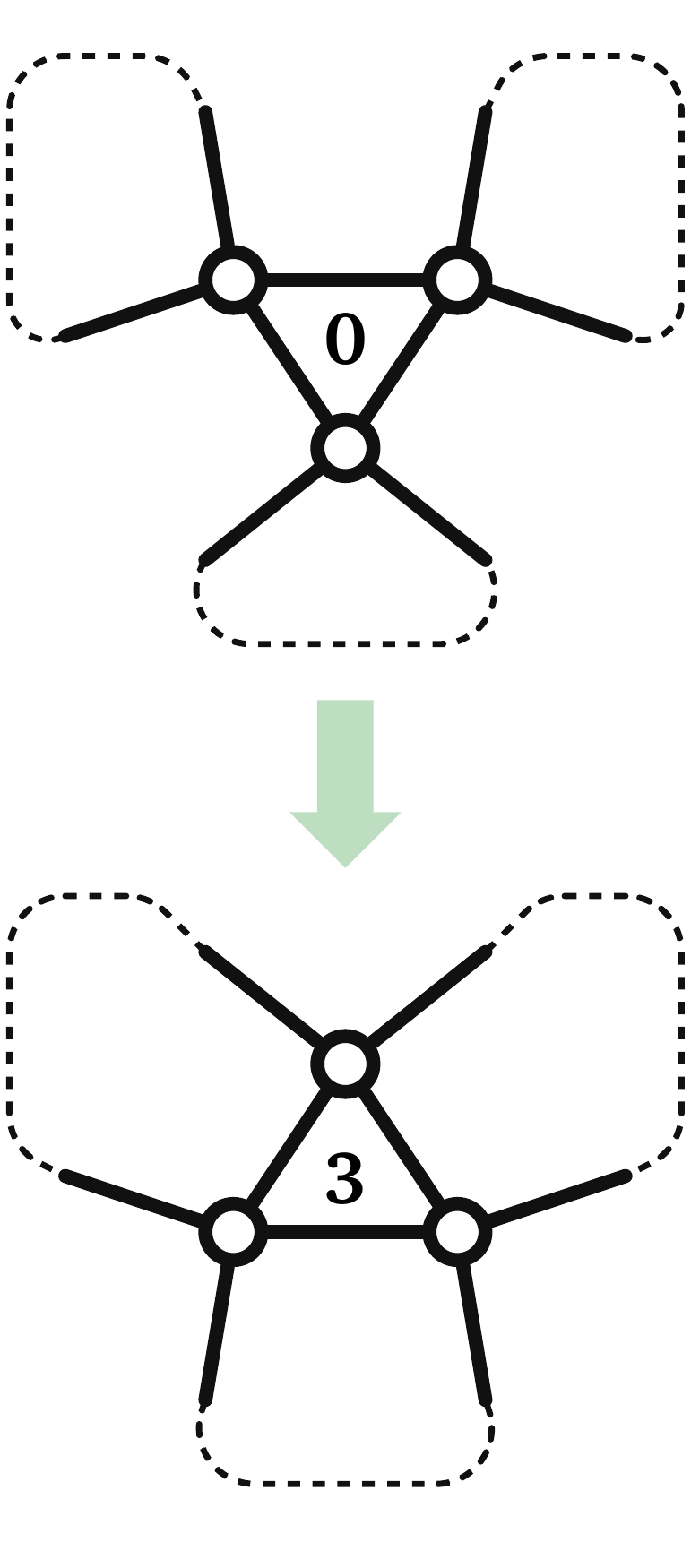}}
	&
	\raisebox{-.5\height}{\includegraphics[scale=0.25]{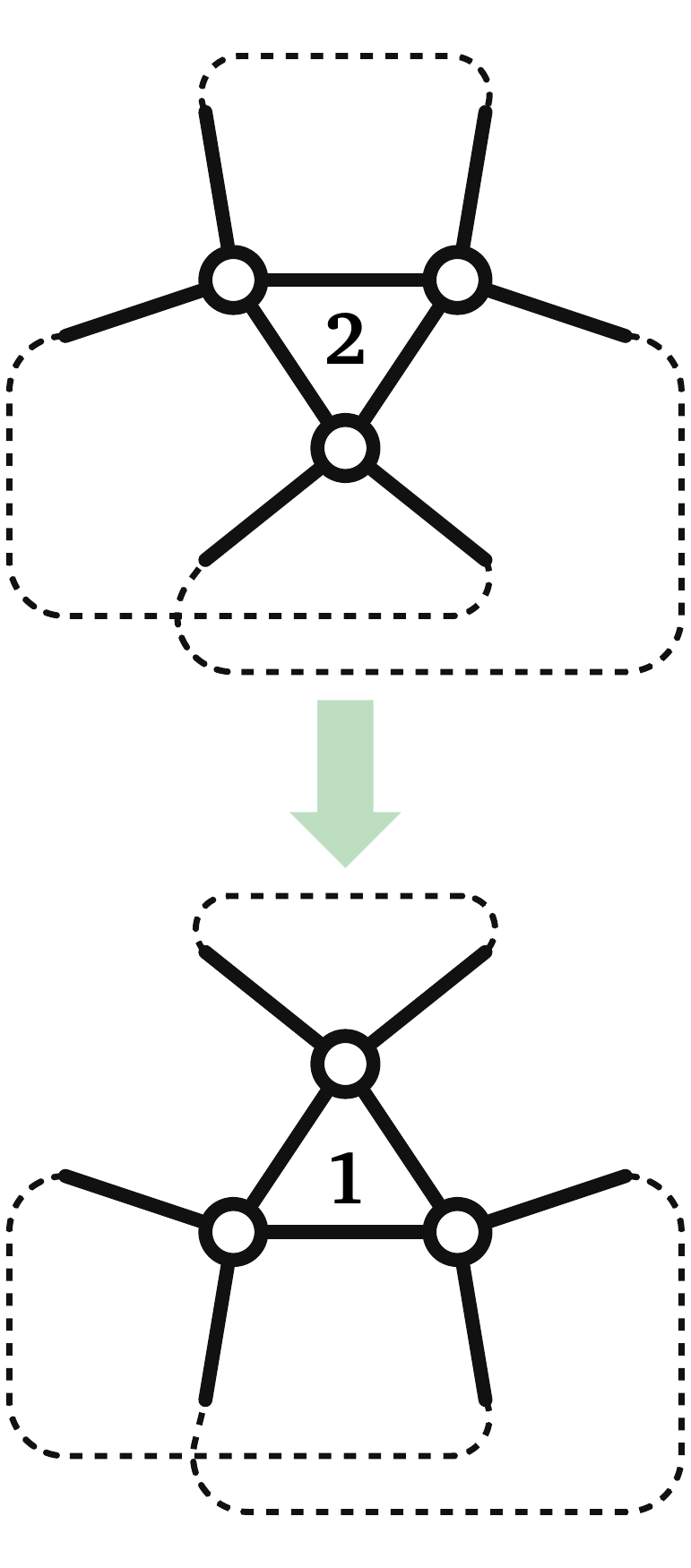}}
\\ \hline
	$δ(γ')-δ(γ)$
	& $0$ & $0$ & $-2$ & $+2$ & $+2$
\end{tabular}
\caption{Changes to the defect incurred by homotopy moves. Numbers indicate how many pairs of vertices in the figure are interleaved; dashed lines indicate the order in which the arcs are traversed.}
\label{F:defect-moves}
\end{figure}

Arnold~\cite{a-tipcc-94, a-pctip-94} originally defined two related curve invariants $\emph{St}$ (“strangeness”) and $J^+$ by their changes under $\arc20$ and $\arc33$ homotopy moves, without giving explicit formulas.  Specifically, $\arc33$ moves change strangeness by $\pm 1$ as shown in Figure \ref{F:defect-moves} but do not affect $J^+$; $\arc20$ moves change $J^+$ by either~$0$ or $2$ as shown in Figure \ref{F:defect} but do not affect strangeness.  Aicardi~\cite{a-tc-94} later proved that the linear combination $2\emph{St} + J^+$ is unchanged under $\arc10$ moves; Arnold dubbed this linear combination “defect”.  Shumakovich~\cite{s-efspc-95,s-efspc-96} proved that the strangeness of an $n$-vertex planar curve lies between $-\floor{n(n-1)/6}$ and $n(n+1)/2$, and that these bounds are exact.  (Nowik's $Ω(n^2)$ lower bound for regular homotopy moves~\cite{n-cpsc-09} follows immediately from Shumakovich's construction.)  However, the curves with extremal strangeness actually have defect zero.

%

\subsection{Flat Torus Knots}
\label{SS:torus-knots}

Following Hayashi \etal~\cite{hh-msrmd-11,hhsy-musrm-12} and Even-Zohar \etal~\cite{ehln-irkl-14}, we now describe an infinite family of curves with absolute defect $Ω(n^{3/2})$. For any relatively prime positive integers $p$ and $q$, let \EMPH{$T(p,q)$} denote the curve with the following parametrization, where $θ$ runs from $0$ to $2π$:
\[
	T(p,q)(θ) = \big((\cos (qθ)+2) \cos(pθ),~ (\cos (qθ)+2) \sin(pθ)\big).
\]
The curve $T(p,q)$ winds around the origin $p$ times, oscillating $q$ times between two concentric circles and crossing itself exactly $(p-1)q$ times.  We call these curves \EMPH{flat torus knots}.  The flat torus knot $T(p,q)$ is also the medial graph of the $\floor{p/2}\times q$ cylindrical grid, with an additional central vertex if $p$ is odd.  For any $p\le \floor{q/2}$, the flat torus knot $T(p,q)$ is isomorphic to the regular star polygon with Schläfli symbol $\set{q/p}$.

\begin{figure}[htb]
\centering
\includegraphics[width=2in]{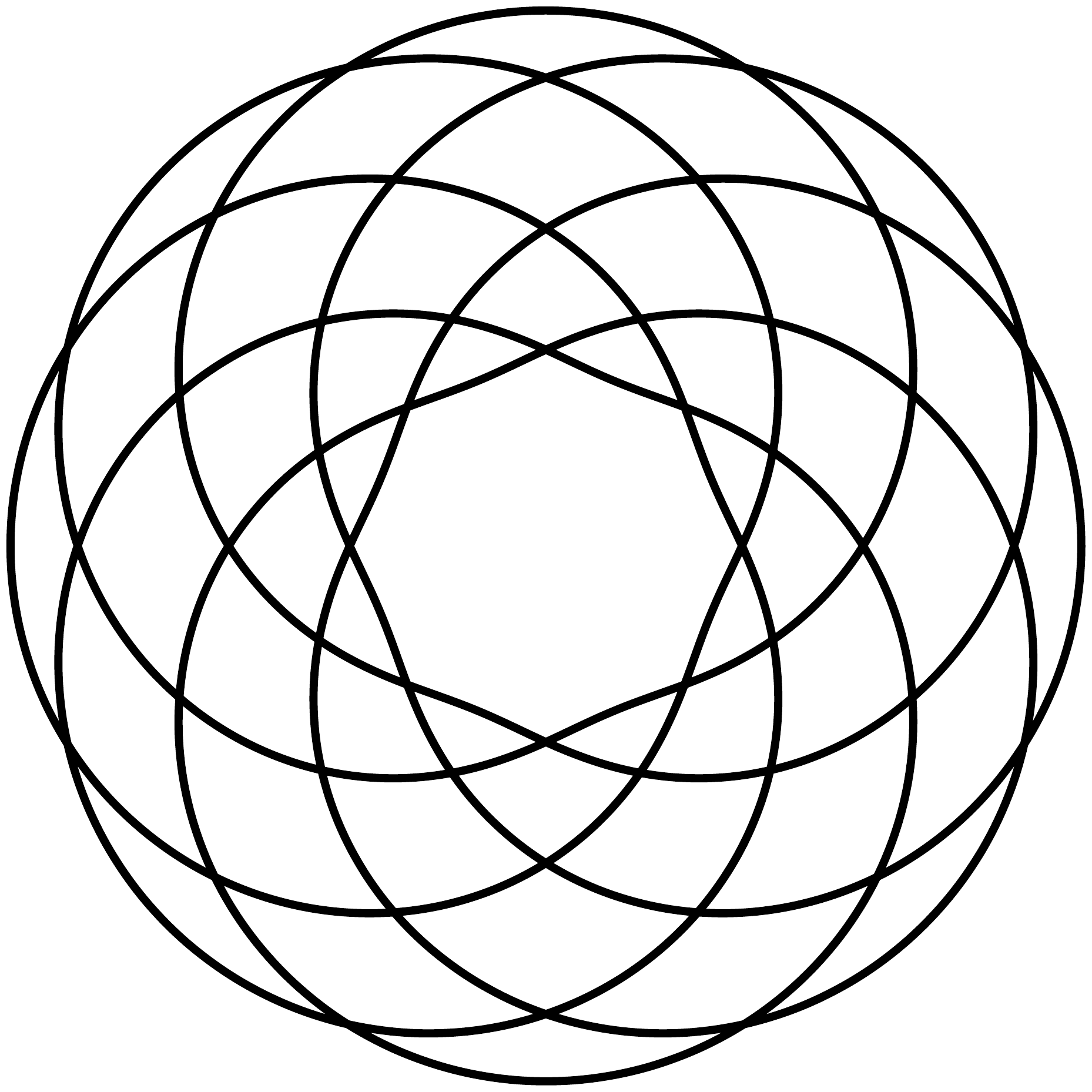}
\qquad
\includegraphics[width=2in]{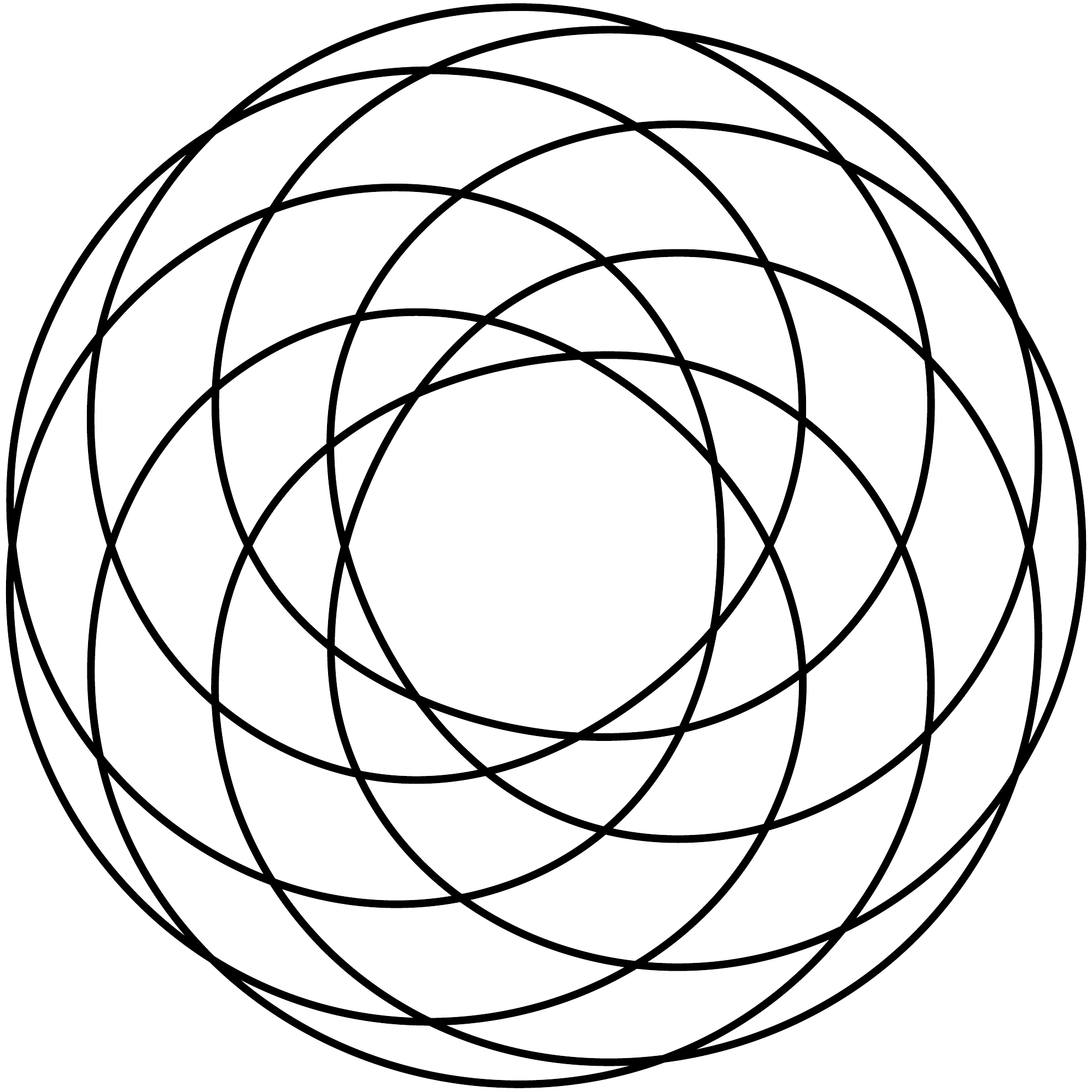}
\caption{The flat torus knots $T(7,8)$ and $T(8,7)$.}
\end{figure}

Hayashi \etal~\cite[Proposition~3.1]{hhsy-musrm-12} proved that for any integer $q$, the flat torus knot $T(q+1,q)$ has defect $-2\binom{q}{3}$.  Even-Zohar \etal~\cite{ehln-irkl-14} used a star-polygon representation of the curve $T(p, 2p+1)$ as the basis for a universal model of random knots; using our terminology and notation, they proved that $δ(T(p, 2p+1)) = 4\binom{p+1}{3}$ for any integer $p$.  Our results in this section simplify and generalize both of these results.

For purposes of illustration, we cut each torus knot $T(p,q)$ open into a “flat braid” consisting of $p$ $x$-monotone paths, which we call \emph{strands}, between two fixed diagonal lines.  All strands are directed from left to right.

\begin{lemma}
\label{L:braid-wide}
$δ(T(p, ap+1)) = 2a \binom{p+1}{3}$ for all positive integers $a$ and $p$.
\end{lemma}

\begin{proof}
The curve $T(p, 1)$ can be reduced using only $\arc10$ moves, so its defect is zero.  For any integer $a\ge 0$, we can reduce $T(p, ap-1)$ to $T(p, ({a-1})p-1)$ by straightening the leftmost block of $p(p-1)$ crossings in the flat braid representation, one strand at a time.  Within this block, each pair of strands in the flat braid intersect twice.  Straightening the bottom strand of this block requires the following $\binom{p}{2}$ moves, as shown in Figure \ref{F:braid-wide}.

\begin{itemize}\itemsep0pt
\item
$\binom{p-1}{2}$ $\arc33$ moves pull the bottom strand downward over one intersection point of every other pair of strands.  Just before each $\arc33$ move, every pair of the three relevant vertices is interleaved, so each move decreases the defect by $2$.

\item
$(p-1)$ $\arc20$ moves eliminate a pair of intersection points between the bottom strand and every other strand.  Each of these moves also decreases the defect by $2$.
\end{itemize}

Altogether, straightening one strand decreases the defect by $\smash{2\binom{p}{2}}$.  Proceeding similarly with the other strands, 
we conclude that $δ(T(p, ap+1)) = δ(T(p, {(a-1)p+1})) + 2\binom{p+1}{3}$.  The lemma follows immediately by induction.
\end{proof}

\begin{figure}[htb]
\centering
\includegraphics[scale=0.33]{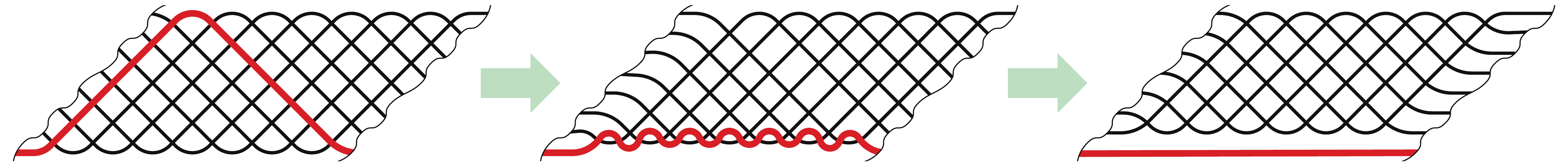}
\caption{Straightening one strand in a block of a wide flat torus knot.}
\label{F:braid-wide}
\end{figure}

\begin{lemma}
\label{L:braid-deep}
$δ(T(aq+1, q)) = -2a \binom{q}{3}$ for all positive integers $a$ and $q$.
\end{lemma}

\begin{proof}
The curve $T(1,q)$ is simple, so its defect is trivially zero.  For any positive integer $a$, we can transform $T(aq+1, q)$ into $T((a-1)q+1, q)$ by incrementally removing the innermost $q$ \emph{loops}.  We can remove the first loop using $\binom{q}{2}$ homotopy moves, as shown in Figure \ref{F:braid-deep}.  (The first transition in Figure~\ref{F:braid-deep} just reconnects the top left and top right endpoints of the flat braid.)
\begin{itemize}
\item
$\binom{q-1}{2}$ $\arc33$ moves pull the left side of the loop to the right, over the crossings inside the loop.  Just before each $\arc33$ move, the three relevant vertices contain two interleaved pairs, so each move \emph{increases} the defect by $2$.
\item
$q-1$ $\arc20$ moves pull the loop over $q-1$ strands.  The strands involved in each move are oriented in opposite directions, so these moves leave the defect unchanged.
\item
Finally, we can remove the loop with a single $\arc10$ move, which does not change the defect.
\end{itemize}

Altogether, removing one loop increases the defect by $\smash{2\binom{q-1}{2}}$.  Proceeding similarly with the other loops, 
we conclude that $δ(T(aq+1, q)) = δ(T((a-1)q+1, q)) - 2 \binom{q}{3}$.  The lemma follows immediately by induction.
\end{proof}

\begin{figure}[htb]
\centering
\includegraphics[scale=0.33]{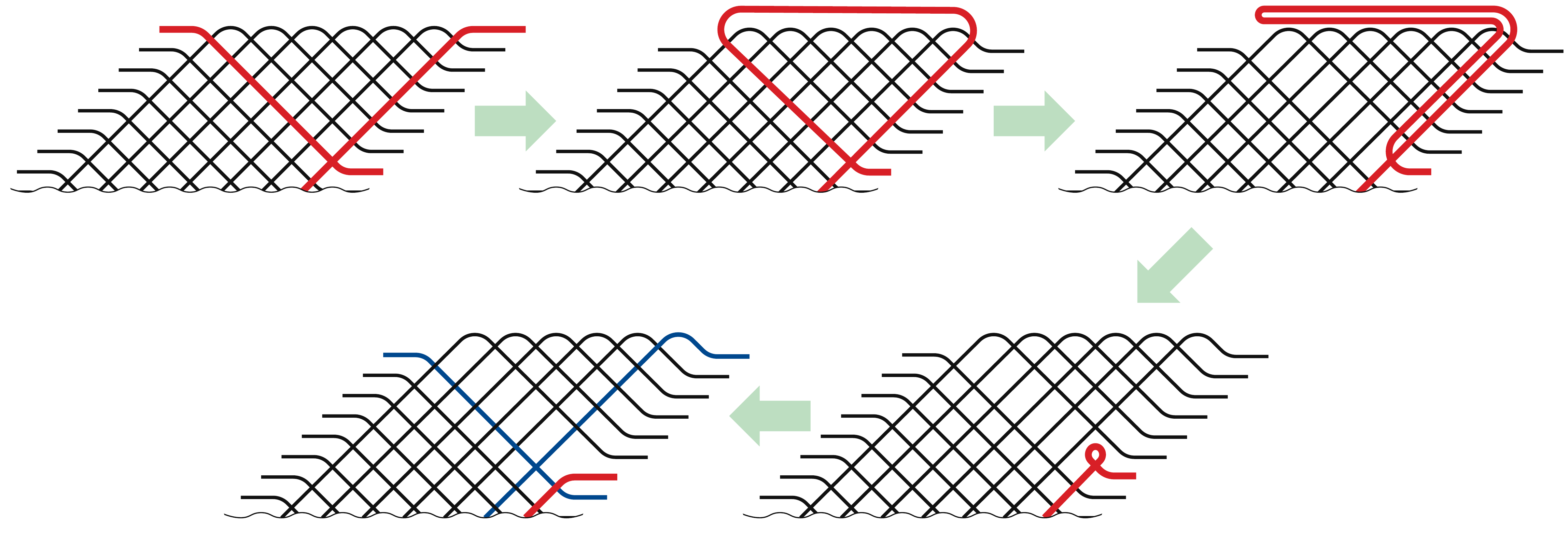}\\[-2ex]
\caption{Removing one loop from the innermost block of a deep flat torus knot.}
\label{F:braid-deep}
\end{figure}

\begin{corollary}
For every positive integer $n$, there are closed curves with $n$ vertices whose defects are  $n^{3/2}/3 - O(n)$ and $-n^{3/2}/3 + O(n)$.
\end{corollary}

\begin{proof}
The lower bound follows from the previous lemmas by setting $a=1$.  If $n$ is a prefect square, then the flat torus knot $T(\sqrt{n}+1, \sqrt{n})$ has $n$ vertices and defect $\smash{-2\binom{\sqrt{n}}{3}}$.  If $n$ is not a perfect square, we can achieve defect ${-2\binom{\floor{\sqrt{n}}}{3}}$ by applying $\arc01$ moves to the curve $T(\floor{\sqrt{n}}+1, \floor{\sqrt{n}})$.  Similarly, we obtain an $n$-vertex curve with defect $\smash{2\binom{\floor{\sqrt{n+1}}+1}{3}}$ by adding loops to the curve $T(\floor{\sqrt{n+1}}, \floor{\sqrt{n+1}}+1)$.
\end{proof}

Corollary \ref{C:lowerbound} now immediately implies the following lower bounds:

\begin{corollary}
For every positive integer $n$, there is a closed curve with $n$ vertices that requires at least $n^{3/2}/6 - O(n)$ homotopy moves to reduce to a simple closed curve.
\end{corollary}

\begin{corollary}
For every positive integer $n$, there is a plane graph with $n$ edges that requires at least $n^{3/2}/6 - O(n)$ electrical transformations to reduce to a single vertex.
\end{corollary}
%

\subsection{Tight Bounds for Grids}

Finally, we derive tight lower bounds on the number of electrical transformations required to reduce arbitrary rectangular or cylindrical grids.  In particular, we show that Truemper's $O(p^3)$ upper bound for the $p\times p$ square grid~\cite{t-drpg-89,t-md-92} and Nakahara and Takahashi's $O(\min\set{pq^2, p^2q})$ upper bound for the $p\times q$ cylindrical grid ~\cite{nt-aafts-96} are both tight up to constant factors. 

\begin{corollary}
\label{C:cylindrical-grid}
For all positive integers $p$ and $q$, the $p\times q$ cylindrical grid requires $Ω(\min\set{p^2 q, p q^2})$ electrical transformations to reduce to a single vertex.
\end{corollary}

\begin{proof}
Let $G$ denote the $p\times q$ cylindrical grid; we need to prove that $X(G^\times) = Ω(\min\set{p^2 q, p q^2})$.

First suppose $2p < q$.  Let $γ$ denote the flat torus knot $T(2p, 2ap+1)$, where $a = \floor{(q-1)/2p}$.  This curve is a connected smoothing of $G^\times$, so Lemma~\ref{L:smoothing} implies $X(G^\times) \ge X(H^\times)$.  Lemma \ref{L:braid-wide} gives us
\[
	δ(γ) = 2a\binom{2p+1}{3} = Ω(ap^3) = Ω(p^2q).
\]
The lower bound $X(H^\times) = Ω(p^2q)$ now follows by Corollary~\ref{C:lowerbound}.

The case $q \le 2p$ is similar.  Let $γ$ denote the flat torus knot $T(2aq+1, q)$, where $a = \floor{p/q}$.   This curve is a connected smoothing of $G^\times$, so Lemma~\ref{L:smoothing} implies $X(G^\times) \ge X(H^\times)$.  Lemma \ref{L:braid-deep} gives us
\[
	\Abs{\strut δ(γ)} = 2a\binom{q}{3} = Ω(aq^3) = Ω(pq^2).
\]
The lower bound $X(H^\times) = Ω(pq^2)$ now follows by Corollary~\ref{C:lowerbound}.
\end{proof}

\begin{corollary}
For all positive integers $p$ and $q$, the $p\times q$ rectangular grid requires $Ω(\min\set{p^2 q, p q^2})$ electrical transformations to reduce to a single vertex.
\end{corollary}

\begin{proof}
The $\floor{p/3}\times\floor{q/3}$ cylindrical grid is a minor of the $p\times q$ rectangular grid, so this lower bound follows from Corollary \ref{C:cylindrical-grid} and 
Lemma \ref{L:smoothing}.
\end{proof}

\begin{corollary}
For every positive integer $t$, every planar graph with treewidth $t$ requires $Ω(t^3)$ electrical transformations to reduce to a single vertex.
\end{corollary}

\begin{proof}
Every planar graph with treewidth $t$ contains an $Ω(t)\times Ω(t)$ grid minor~\cite{rst-qepg-94}, so this lower bound also follows from Corollary \ref{C:cylindrical-grid} and Lemma \ref{L:smoothing}.
\end{proof}

\section{Defect Upper Bound}
\label{S:upper-bound}

\def\Simplify{\mathbin{\circledcirc}}
\def\SimpAbove{\mathbin\Cap}
\def\SimpBelow{\mathbin\Cup}

In this section, we prove a matching $O(n^{3/2})$ upper bound on the absolute value of the defect, using a recursive inclusion-exclusion argument. Throughout this section, let $γ$ be an arbitrary non-simple closed curve, and let $n$ be the number of vertices of $γ$.

\subsection{Winding Numbers and Diameter}
\label{SS:diameter}

\def\Loopx{γ\!_x}
\def\Beforex{α_x}
\def\Afterx{ω_x}

First we derive an upper bound in terms of the diameter of the dual graph \EMPH{$γ^*$}.  (If $γ$ is the medial graph of a plane graph $G$, then $γ^*$ is the vertex-face incidence graph of $G$, otherwise known as the \emph{radial graph} of $G$.)  The upper bound $\Abs{δ(γ)} = O(n\cdot\Diam(γ^*))$ follows from Corollary \ref{C:lowerbound} and the electrical reduction algorithm of Feo and Provan~\cite{fp-dtert-93}; here we give a simpler direct proof.

We parametrize $γ$ as a function $γ\colon [0,1]\to\Real^2$, where $γ(0) = γ(1)$ is an arbitrarily chosen basepoint.  For each vertex $x$ of $γ$, let \EMPH{$\Loopx$} denote the closed subpath of $γ$ from the first occurrence of $x$ to the second. More formally, if $x = γ(u) = γ(v)$ where $0 < u < v < 1$, then $\Loopx$ is the closed curve defined by setting $\Loopx(t) \coloneqq γ((1-t)u + tv)$ for all $0 ≤ t ≤ 1$.

\begin{lemma}
For every vertex $x$, we have $\sum_{y\between x} \sgn(y) = 2\Wind(\Loopx, x) - 2\Wind(\Loopx, γ(0)) - \sgn(x)$.
\label{L:subwind}
\end{lemma}

\begin{proof}
Our proof follows an argument of Titus~\cite[Theorem 1]{t-tncsa-60}.

Fix a vertex $x = γ(u) = γ(v)$, where $0<u<v<1$. Let \EMPH{$\Beforex$} denote the subpath of $γ$ from $γ(0)$ to $γ(u-ε)$, and let \EMPH{$\Afterx$} denote the subpath of $γ$ from $γ(v+ε)$ to $γ(1) = γ(0)$, for some sufficiently small $ε>0$.  Specifically, we choose $ε$ such that there are no vertices $γ(t)$ where $u-ε≤t<u$ or $v<t≤v+ε$.  (See Figure \ref{F:subwind}.)  A vertex $y$ interleaves with $x$ if and only if $y$ is an intersection point of $\Loopx$ with either~$\Beforex$ or $\Afterx$, so
\[
	\sum_{y\between x} \sgn(y)
	~=
	\sum_{\!\!\!y \in \Beforex \cap \Loopx \vphantom{\between}\!\!\!} \sgn(y)
		~+
	\sum_{\!\!\!y \in \Loopx \cap \Afterx \vphantom{\between}\!\!\!} \sgn(y).
\]

\begin{figure}[htb]
\centering
\includegraphics[scale=0.4]{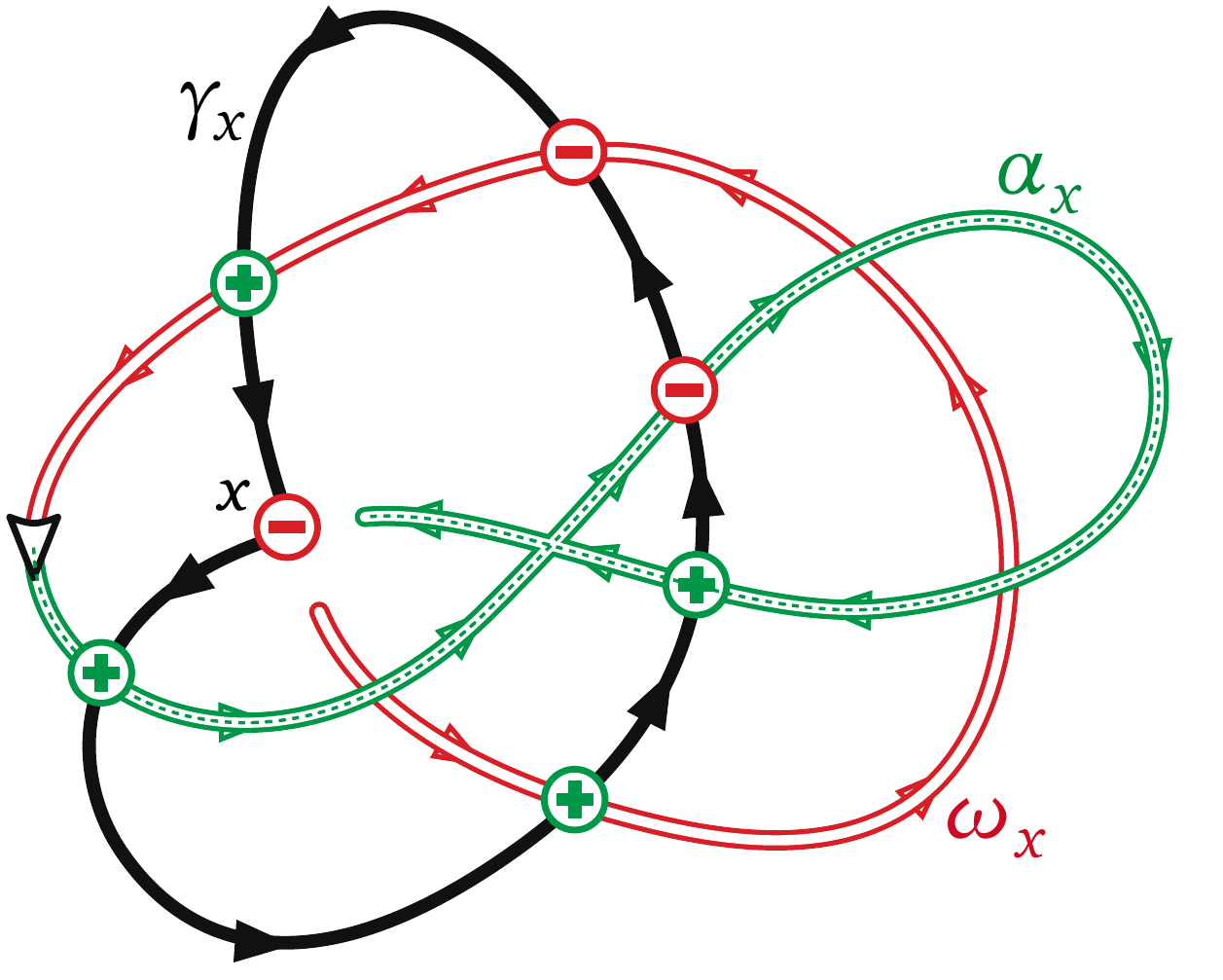}
\caption{Proof of Lemma \ref{L:subwind}: $\Wind(\Loopx, x) = +1-1+1-\tfrac{1}{2} = \tfrac12$}
\label{F:subwind}
\end{figure}

Now suppose we move a point $p$ continuously along the path $α_x$, starting at the basepoint $γ(0)$.  The winding number $\Wind(\Loopx, p)$ changes by $1$ each time this point $\Loopx$.  Each such crossing happens at a vertex of $γ$ that lies on both $\Beforex$ and~$\Loopx$; if this vertex is positive, $\Wind(\Loopx, p)$ increases by $1$, and if this vertex is negative,  $\Wind(\Loopx, p)$ decreases by $1$.  It follows that
\[
	\sum_{\!\!\!y \in \Beforex \cap \Loopx \vphantom{\between}\!\!\!} \sgn(y)
	~=~
	\Wind(\Loopx, γ(u-ε)) - \Wind(\Loopx, γ(0)).
\]
Symmetrically, if we move a point $p$ \emph{backward} along $\Afterx$ from the basepoint, the winding number $\Wind(\Loopx, p)$ increases (resp. decreases) by $1$ whenever $γ(t)$ passes through a positive (resp. negative) vertex in $\Afterx \cap \Loopx$; see the red path in Figure \ref{F:subwind}. Thus, 
\[
	\sum_{\!\!\!y \in \Afterx \cap \Loopx \vphantom{\between}\!\!\!} \sgn(y)
	~=~
	\Wind(\Loopx, γ(v+ε)) - \Wind(\Loopx, γ(0)).
\]
Finally, our sign convention for vertices implies
\[
	\Wind(\Loopx, γ(u-ε))
	~=~ \Wind(\Loopx, γ(v+ε))
	~=~ \Wind(\Loopx, x) - \sgn(x)/2,
\]
which completes the proof.
\end{proof}

\begin{lemma}
\label{L:diameter}
For any closed curve $γ$, we have $\Abs{δ(γ)} ≤ 2n\cdot\Diam(γ^*) + n$.
\end{lemma}

\begin{proof}
Polyak's defect formula can be rewritten as
\[
	δ(γ) 	= -\sum_x \sgn(x) \left(\sum_{y\between x} \sgn(y) \right).
\]
(This sum actually considers every pair of interleaved vertices twice, which is why the factor $2$ is omitted.)  Assume without loss of generality that the basepoint $γ(0)$ lies on the outer face of $γ$, so that $\Wind(\Loopx, γ(0)) = 0$ for every vertex $x$. Then the previous lemma implies
\[
	δ(γ) 	~=~ \sum_x \sgn(x) \left(\sgn(x) - 2\Wind(\Loopx, x) \strut \right)
		~=~ n - 2\sum_x \sgn(x) \cdot \Wind(\Loopx, x),
\]
and therefore
\[
	\Abs{δ(γ)} \le n + 2\sum_x \Abs{\Wind(\Loopx, x)}.
\]
We easily observe that $\Abs{\Wind(\Loopx, x)} \le \Diam(\Loopx^*) ≤ \Diam(γ^*)$ for every vertex $x$; the second inequality follows from the fact that no path crosses $\Loopx$ more times than it crosses $γ$.  The lemma now follows immediately.
\end{proof}
%
%

\subsection{Inclusion-Exclusion}
\label{SS:inex}

Let $σ$ be an arbitrary \emph{simple} closed curve that intersects $γ$ only transversely and away from its vertices. Let $s$ be the number of intersection points between $γ$ and $σ$; the Jordan curve theorem implies that $s$ must be even. Let $z_0, z_1, \dots, z_{s-1}$ be the points in $σ\cap γ$ in order along $γ$ (\emph{not} in order along $σ$). These intersection points decompose $γ$ into a sequence of $s$ subpaths $γ_1, γ_2, \dots, γ_s$; specifically, $γ_i$ is the subpath of $γ$ from $z_{i-1}$ to $z_{i\bmod s}$, for each index $i$. Without loss of generality, every odd-indexed path $γ_{2i+1}$ lies outside $σ$, and every even-indexed path $γ_{2i}$ lies inside $σ$.  

Let \EMPH{$γ\SimpBelow σ$} denote a regular curve obtained from $γ$ by continuously deforming all subpaths~$γ_i$ outside~$σ$, keeping their endpoints fixed and never moving across $σ$, to minimize the number of intersections. (There may be several curves that satisfy the minimum-intersection condition; choose one arbitrarily.) Similarly, let \EMPH{$γ\SimpAbove σ$} denote any regular curve obtained by continuously deforming the subpaths $γ_i$ inside~$σ$ to minimize intersections. Finally, let \EMPH{$γ \Simplify σ$} denote the curve obtained by deforming \emph{all} subpaths $γ_i$ to minimize intersections; in other words, $γ \Simplify σ \coloneqq (γ\SimpAbove σ)\SimpBelow σ = (γ\SimpBelow σ)\SimpAbove σ$.  See Figure \ref{F:simplified}.

\begin{figure}[htb]
\centering
\includegraphics[scale=0.3]{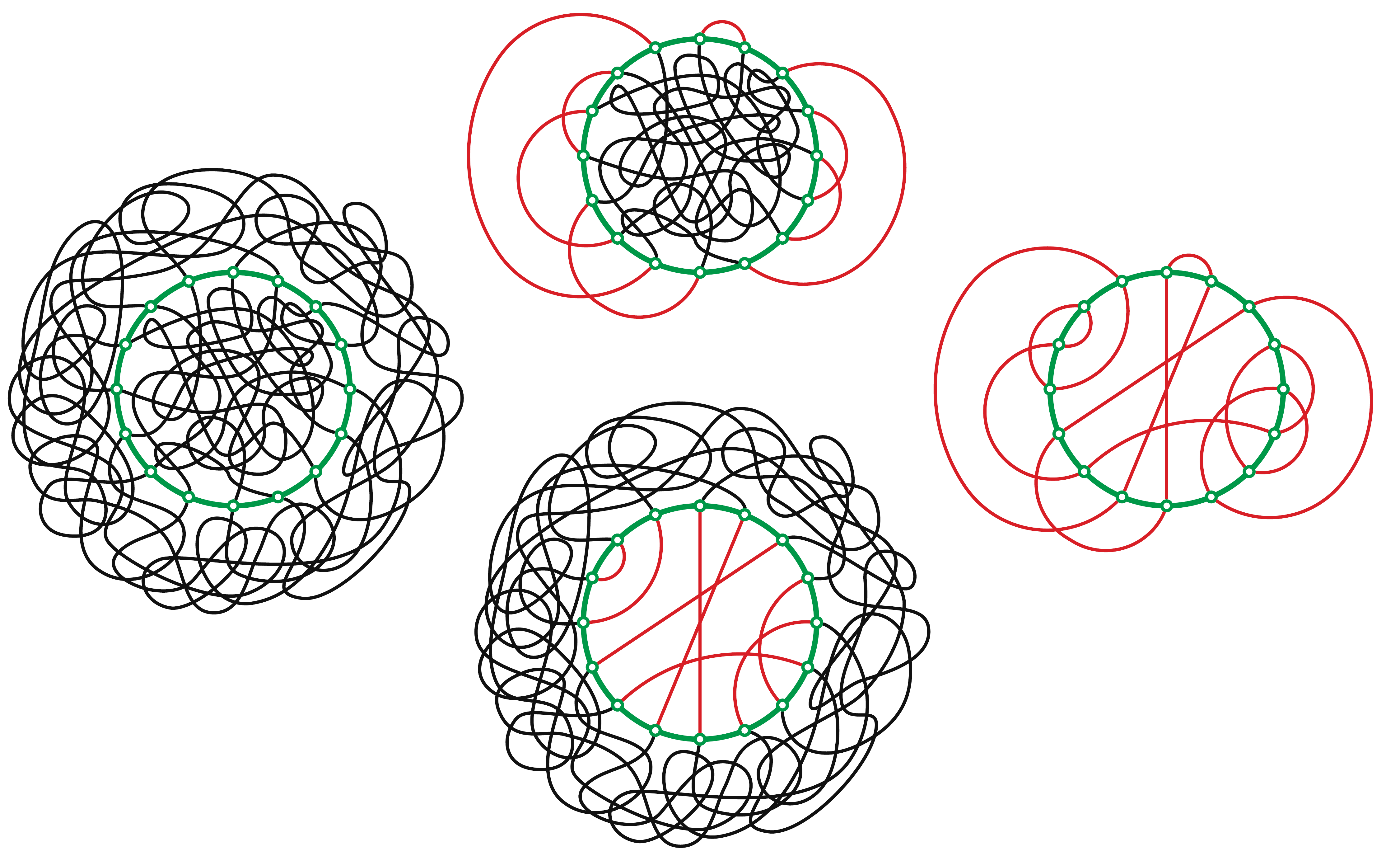}
\caption{Clockwise from left: $γ$, $γ\SimpBelow σ$, $γ\Simplify σ$, and $γ\SimpAbove σ$.  The green circle in all four figures is $σ$.}
\label{F:simplified}
\end{figure}

To simplify notation, we define
\[
	\EMPH{$δ(x,y)$} ~\coloneqq~ [x\between y] \cdot \sgn(x)\cdot\sgn(y)
\]
for any two vertices $x$ and $y$, where $[x\between y] = 1$ if $x$ and $y$ are interleaved and $[x\between y] = 0$ otherwise. Then we can define the defect of $γ$ as
\[
	δ(γ) = -2 \sum_{x,y} δ(x,y).
\]

Every vertex of $γ$ lies at the intersection of two (not necessarily distinct) subpaths. For any index $i$, let $X(i,i)$ denote the set of self-intersection points of $γ_i$, and for any indices $i<j$, let $X(i,j)$ be the set of points where $γ_i$ intersects $γ_j$.

If two vertices $x \in X(i,k)$ and $y \in X(j,l)$ are interleaved, then we must have $i ≤ j ≤ k ≤ l$.  Thus, we can express the defect of $γ$ in terms of crossings between subpaths $γ_i$ as follows.
\[
	δ(γ) = -2 \sum_{i ≤ j ≤ k ≤ l}~
				\sum_{x\in X(i,k)}~
					\sum_{y\in X(j,l)}~
						δ(x,y)
\]

On the other hand, if $i<j<k<l$, then every vertex $x \in γ_i \cap γ_k$ is interleaved with every vertex of $y \in γ_j \cap γ_l$.  Thus, we can express the contribution to the defect from pairs of vertices on four \emph{distinct} subpaths as follows:
\[
	\EMPH{$δ^{\#}(γ, σ)$} \coloneqq -2 \sum_{i < j < k < l}
				~\sum_{x\in X(i,k)} 
				~\sum_{y\in X(j,l)} 
				~\sgn(x)\cdot\sgn(y)
\]
We can express this function more succinctly as 
\[
	δ^{\#}(γ,σ)
	= -2 \sum_{i<j<k<l}
		δ(i,k)\cdot δ(j,l)
\]
by defining
\[
	\EMPH{$δ(i,j)$} \coloneqq \sum_{x\in X(i,j)} \sgn(x)
\]
for all indices $i < j$.

The following lemma implies that continuously deforming the subpaths $γ_i$ without crossing $σ$ leaves the value $δ^{\#}(γ,σ)$ unchanged, even though such a deformation may change the defect $δ(γ)$.

\begin{lemma}
\label{L:homotopy}
The value $δ(i,j)$ depends only on the parity of $i+j$ and the cyclic order of the endpoints of $γ_i$ and $γ_j$ around $σ$.
\end{lemma}

\begin{proof}
There are only three cases to consider.

If $i+j$ is odd, then $γ_i$ and $γ_j$ lie on opposite sides of $σ$ and therefore do not intersect, so $δ(i,j)=0$. For all other cases, $i+j$ is even, which implies without loss of generality that $j ≥ i+2$.

Suppose the endpoints of $γ_i$ and $γ_j$ do not alternate in cyclic order around $σ$, or equivalently, that the corresponding subpaths of $γ\Simplify σ$ are disjoint.  The Jordan curve theorem implies that there must be equal numbers of positive and negative intersections between $γ_i$ and~$γ_j$, and therefore $δ(i,j) = 0$.

Finally, suppose the endpoints of $γ_i$ and $γ_j$ alternate in cyclic order around $σ$, or equivalently, that the corresponding subpaths of $γ\Simplify σ$ intersect exactly once. Then $δ(i,j) = 1$ if the endpoints $z_i, z_j, z_{i-1}, z_{j-1}$ appear in clockwise order around~$σ$ and $δ(i,j) = -1$ otherwise.
\end{proof}

Now consider an interleaved pair of vertices $x\in X(i,k)$ and $y\in X(j,l)$ where at least two of the indices $i,j,k,l$ are equal.  Trivially, $i$ and $k$ have the same parity, and $j$ and $l$ also have the same parity.  If $i=j$ or $i=l$ or $j=k$ or $j=l$, then all four indices have the same parity.  If $i=k$, then we must also have $i=j$ or $i=l$ (or both), so again, all four indices have the same parity.  We conclude that $x$ and $y$ are either both inside $σ$ or both outside $σ$.

\begin{lemma}
\label{L:inclusion-exclusion}
For any regular closed curve $γ$ and any simple closed curve $σ$ that intersects $γ$ only transversely and away from its vertices, we have $δ(γ) = δ(γ\SimpAbove σ) + δ(γ\SimpBelow σ) - δ(γ \Simplify σ)$.
\end{lemma}

\begin{proof}
Let us write $δ(γ) = δ^{\#}(γ, σ) + δ^{\uparrow}(γ, σ) + δ^{\downarrow}(γ, σ)$, where
\begin{itemize}\itemsep0pt
\item $δ^{\#}(γ, σ)$ considers pairs of vertices on four different subpaths $γ_i$, as above,
\item $δ^{\uparrow}(γ, σ)$ considers pairs of vertices outside $σ$ on at most three different subpaths $γ_i$, and
\item $δ^{\downarrow}(γ, σ)$ considers pairs of vertices inside $σ$ on at most three different subpaths $γ_i$.
\end{itemize}
Lemma~\ref{L:homotopy} implies that
\[
	δ^{\#}(γ, σ)
	~=~ δ^{\#}(γ\SimpAbove σ, σ)
	~=~ δ^{\#}(γ\SimpBelow σ, σ)
	~=~ δ^{\#}(γ \Simplify σ, σ).
\]
The definitions of $γ\SimpAbove σ$ and $γ\SimpBelow σ$ immediately imply the following:
\begin{align*}
	δ^{\uparrow}(γ\SimpAbove σ, σ) & = δ^{\uparrow}(γ \Simplify σ, σ)
	&
	δ^{\downarrow}(γ\SimpAbove σ, σ) &= δ^{\downarrow}(γ, σ) 
	\\
	δ^{\uparrow}(γ\SimpBelow σ, σ) & = δ^{\uparrow}(γ, σ)
	&
	δ^{\downarrow}(γ\SimpBelow σ, σ) &= δ^{\downarrow}(γ \Simplify σ, σ)
\end{align*}
The lemma now follows from straightforward substitution.
\end{proof}

\begin{lemma}
\label{L:sep-cubed}
For any closed curve $γ$ and any simple closed curve $σ$ that intersects $γ$ only transversely and away from its vertices, we have $\Abs{δ(γ\Simplify σ)} = O(\abs{γ\cap σ}^3)$.
\end{lemma}

\begin{proof}
Fix an arbitrary reference point $z \in σ\setminus γ$. For any point $p$ in the plane, there is a path from $p$ to~$z$ that crosses $γ\Simplify σ$ at most $O(s)$ times. Specifically, move from $p$ to the nearest point on $γ\Simplify σ$, then follow $γ\Simplify σ$ to $σ$, and finally follow $σ$ to the reference point $z$. It follows that $\Diam((γ\Simplify σ)^*) = O(s)$.  The curve $γ\Simplify σ$ has at most $\smash{2\binom{s/2}{2}} = O(s^2)$ vertices.  The bound $\Abs{δ(γ\Simplify σ)} = O(s^3)$ now immediately follows from Lemma \ref{L:diameter}.
\end{proof}

\subsection{Divide and Conquer}
\label{SS:recurse}

We call a simple closed curve $σ$ \EMPH{useful} for $γ$ if $σ$ intersects $γ$, but only transversely and away from the vertices of $γ$, and there are at least $\abs{γ\cap σ}^2$ vertices of $γ$ on both sides of $σ$. In particular, a simple closed curve that is disjoint from $γ$ is not useful.
\footnote{We could define a transverse cycle to be useful if there are at least $\alpha\cdot \abs{γ\cap σ}^2$ vertices on both sides, and then optimize $α$ to minimize the resulting upper bound.%
}

\begin{lemma}
\label{L:useful}
If no simple closed curve is useful for $γ$, then $\Diam(γ^*) = O(\sqrt{n})$.
\end{lemma}

\begin{proof}
To simplify notation, let $D = \Diam(γ^*)$. Let $a$ and $z$ be any two points in $\Real^2\setminus γ$ such that any path from $a$ to $z$ crosses $γ$ at least $D$ times.  We construct a nested sequence $σ_1, σ_2, \cdots, σ_D$ of disjoint simple closed curves as follows. For each integer $j$, let $R_j$ denote the set of all points reachable from $a$ by a path that crosses $γ$ less than $j$ times, and let $\tilde{R}_j$ be an arbitrarily small open neighborhood of the closure of $R_j \cup \tilde{R}_{j-1}$. For all $1\le j\le D$, the boundary of~$\tilde{R}_j$ is the disjoint union of simple closed curves, each of which intersect $γ$ transversely away from its vertices (or not at all).  Let $σ_j$ be the unique boundary component of $\tilde{R}_j$ that separates $a$ and $z$.  See Figure~\ref{F:curve-levels}.

\begin{figure}[htb]
\centering
\includegraphics[scale=0.4]{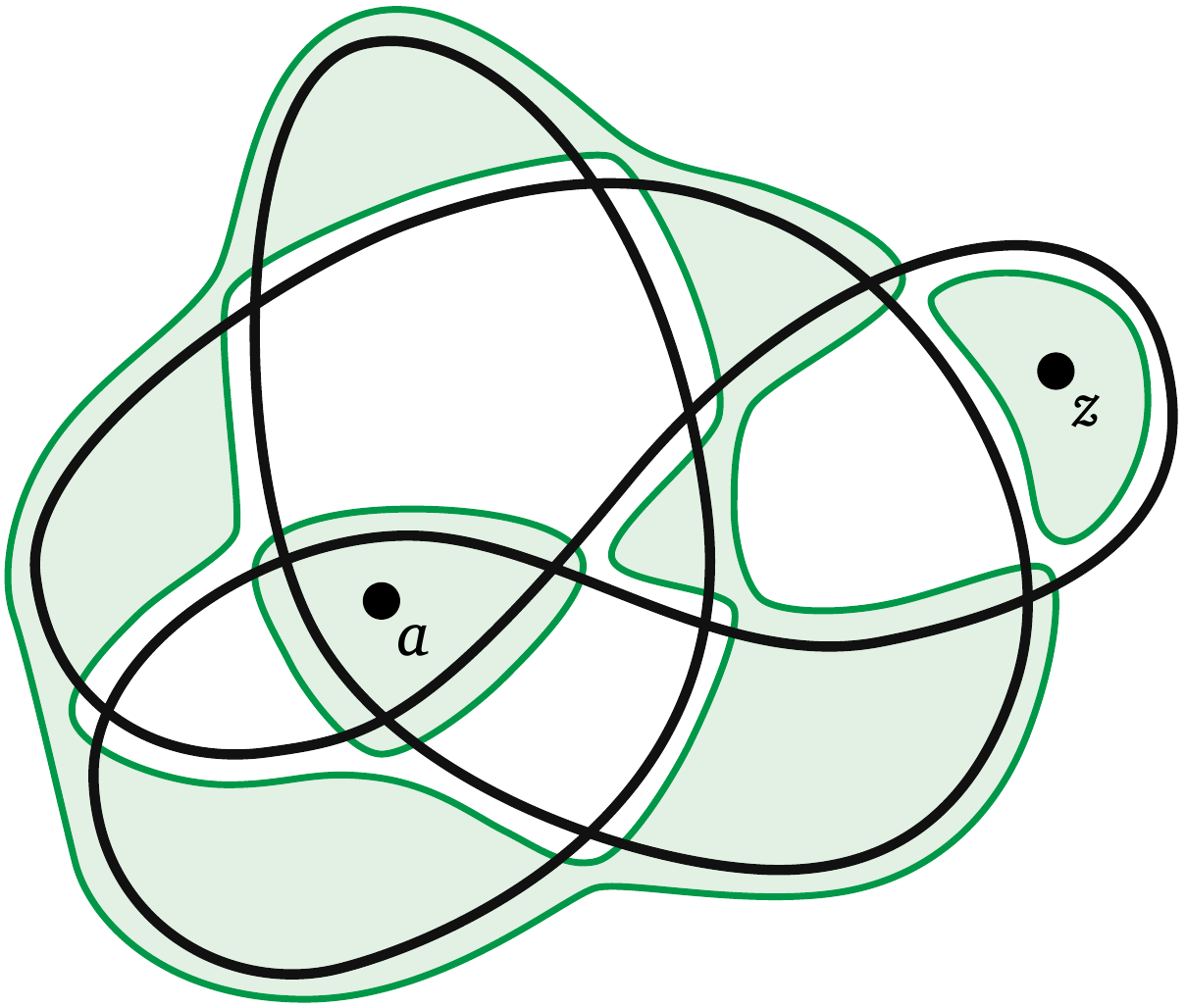}
\caption{Nested simple closed curves transverse to $γ$.}
\label{F:curve-levels}
\end{figure}

For each index $j$, let $A_j$ denote the component of $\Real^2\setminus \tilde{R}_j$ that contains the point $z$; let $n_j$ denote the number of vertices of $γ$ that lie in $A_j$; and  let $s_j = \abs{γ \cap σ_j}$. For notational convenience, we define $A_0 = \varnothing$ and thus $n_0=0$. Finally, let $L$ be the largest index such that $n_L ≤ n/2$; without loss of generality, we can assume $L ≥ D/2$. To prove the lemma, it suffices to show that if no curve $σ_j$ with $j ≤ L$ is useful, then $L = O(\sqrt{n})$.

If an edge of $γ$ crosses $σ_j$, then at least one of its endpoints lies in the annulus $A_j\setminus A_{j-1}$. Moreover, if the edge crosses $σ_j$ twice, then both endpoints lie in $A_j\setminus A_{j-1}$.  Exhaustive case analysis implies that $n_j ≥ n_{j-1} + s_j/2$, and therefore by induction
\[
	n_j \ge \frac{1}{2} \sum_{i < j} s_i,
\]
for every index $j>1$. Trivially, $n_1 \ge 1$ unless $γ$ is simple, and $s_1\ge 2$ unless $D = 1$.

Now suppose no curve $σ_j$ with $1≤ j ≤ L$ is useful. Then we must have $s_j^2 > n_j$ and therefore 
\[
	s_j^2 > \frac{1}{2} \sum_{i < j} s_i
\]
for all $j$. An easy induction argument implies that $s_j > j/5$, and therefore
\[
	\frac{n}{2}
	~≥~ n_L
	~>~ \frac{1}{2} \sum_{i < L} \frac{i}{5}
	~=~ \frac{1}{10} \binom{L}{2}.
\]
We conclude that $L \le \sqrt{10n}$.%
\end{proof}


We are now finally ready to prove our main upper bound.

\begin{theorem}
\label{Th:defect}
$\Abs{δ(γ)} = O(n^{3/2})$ for every generic closed curve $γ$ with $n$ vertices.
\end{theorem}

\begin{proof}
Let $Δ(n)$ denote the maximum absolute defect of any closed curve with $n$ vertices. We prove by induction on $n$ that $Δ(n) ≤ C \cdot n^{3/2}$, for some absolute constant $C$ to be determined.  (The base case $Δ(1) = 2$ implies that $C ≥ 2$.) 

Let $γ$ be an arbitrary closed curve with $n$ vertices. Let $σ$ be a simple closed curve that is useful for $γ$. (If there are no useful curves, then Lemmas \ref{L:diameter} and \ref{L:useful} imply that $\Abs{δ(γ)} = O(n^{3/2})$.) Let $s = \abs{γ\cap σ}$. Lemmas \ref{L:inclusion-exclusion} and \ref{L:sep-cubed} imply 
\[
	\Abs{δ(γ)} = \Abs{δ(γ\SimpAbove σ)} + \Abs{δ(γ\SimpBelow σ)} + O(s^3).
\]
Suppose $m$ vertices of $γ$ lie in the interior of $σ$; without loss of generality, we can assume $m ≤ n/2$. The curve $γ\SimpBelow σ$ has at most $m + \binom{s/2}{2} < m + s^2/8$ vertices; symmetrically, the curve $γ\SimpAbove σ$ has at most $n - m + s^2/8$ vertices.  It follows that
\[
	\Abs{δ(γ)} ~ ≤ ~ Δ(m + s^2/8) + Δ(n - m + s^2/8) + c\cdot s^3
\]
for some constant $c$.

Because $σ$ is useful, we have $m ≥ s^2$, so both arguments of $Δ$ on the right side of this inequality are smaller than $n$.  Thus, the inductive hypothesis implies
\[
	\Abs{δ(γ)} ~ ≤ ~C\left(m + s^2/8\right)^{3/2}
				+ C\left(n - m + s^2/8\right)^{3/2}
				+ c\cdot s^3.
\]

For any fixed $s$, the convexity of the function $x\mapsto x^{3/2}$ implies that right side of this inequality is maximized when $m = s^2$, so
\[
	\Abs{δ(γ)} ~ ≤ ~ C\left(9s^2/8 \right)^{3/2}
				+ C\left(n - 7s^2/8\right)^{3/2}
				+ c\cdot s^3.
\]
The inequality $(x - y)^{3/2} ≤ (x-y)x^{1/2} = x^{3/2} - yx^{1/2}$ now implies
\[
	\Abs{δ(γ)} ~≤~ Cn^{3/2} + C\left(9s^2/8 \right)^{3/2}
				- C\left(7s^2/8\right)n^{1/2}
				+ c\cdot s^3.
\]
Finally, because $σ$ is useful, we must have $\sqrt{n} \ge \sqrt{2} \cdot s$, which implies
\begin{align*}
	\Abs{δ(γ)}
	& ~≤~ Cn^{3/2} + C
		\left(	\left(9/8 \right)^{3/2} - 7\sqrt{2}/8	\right) s^3 + c\cdot s^3
\\	& ~=~ Cn^{3/2} - \left(\sqrt{2}C/32 - c\right) s^3.
\end{align*}
Provided $C/c > 16 \sqrt{2}$, then $\Abs{δ(γ)} ≤ Cn^{3/2}$, as required.%
\footnote{A more careful analysis to Lemma~\ref{L:sep-cubed} implies that $c < 3/4$, and therefore it suffices to set $C = 12\sqrt{2}$.  One can further reduce the constant $C$ by redefining a transverse cycle to be useful if there are at least $\alpha\cdot \abs{γ\cap σ}^2$ vertices on both sides, and then optimize $\alpha$ to obtain a refined bound on $C$.  We get $C < 5.6429$.}
\end{proof}

\subsection{Implications for Random Knots}

Finally, we describe some interesting implications of our results on the expected behavior of  random knots, following earlier results of Lin and Wang~\cite{lw-igopc-96}, Polyak~\cite{p-icfgd-98} and Even-Zohar \etal~\cite{ehln-irkl-14}.  We refer the reader to Burde and Zieschang~\cite{bz-k-03} or Kauffman~\cite{k-k-87} for further background on knot theory, and to Chmutov \etal~\cite{cdm-ivki-12} for a detailed overview of finite-type knot invariants; we include only a few elementary definitions to keep the paper self-contained.

A \EMPH{knot} is (the image of) a continuous injective map from the circle into $\Real^3$.  Two knots are considered equivalent (more formally, \emph{ambient isotopic}) if there is a continuous deformation of $\Real^3$ that deforms one knot into the other.  Knots are often represented by \EMPH{knot diagrams}, which are 4-regular plane graphs defined by a generic projection of the knot onto the plane, with an annotation at each vertex indicating which branch of the knot is “over” or “under” the other.  Call any crossing $x$ in a knot diagram \emph{ascending} if the first branch through $x$ after the basepoint passes over the second, and \emph{descending} otherwise.

The \EMPH{Casson invariant $c_2$} is the simplest finite-type knot invariant; it is also equal to the second coefficient of the Conway polynomial~\cite{bl-kpvi-93,pv-cki-01}.  Polyak and Viro~\cite{pv-gdfvi-94,pv-cki-01} derived the following combinatorial formula for the Casson invariant of a knot diagram~$κ$:
\[
	c_2(κ) ~=~ - \!\!\!\sum_{\text{descending $x$}}~
			\sum_{\text{ascending $y$}}~
			[x\between y] \cdot \sgn(x) \cdot \sgn(y).
\]
Like defect, the value of $c_2(κ)$ is independent of the choice of basepoint or orientation of the underlying curve $γ$; moreover, if the knots represented by diagrams $κ$ and $κ'$ are equivalent, then $c_2(κ) = c_2(κ')$.

Polyak~\cite[Theorem 7]{p-icfgd-98} observed that if a knot diagram $κ$ is obtained from an arbitrary closed curve~$γ$ by independently resolving each crossing as ascending or descending with equal probability, then one can relate the expectation of Casson invariant $c_2(κ)$ and the defect of $γ$ by
\[
	\E[c_2(κ)] = δ(γ)/8.
\]
The same observation is implicit in earlier results of Lin and Wang~\cite{lw-igopc-96}; and (for specific curves) in the later results of Even-Zohar \etal~\cite{ehln-irkl-14}.

Even-Zohar \etal~\cite{ehln-irkl-14} studied the distribution of the Casson invariant for two models of random knots, the \emph{Petaluma} model of Adams \etal~\cite{a-tcnkl-13, acdll-kpwsm-15}, which uses singular one-vertex diagrams consisting of $2p+1$ disjoint non-nested loops for some integer $p$, and the \emph{star} model, which uses (a polygonal version of) the flat torus knot $T(p, 2p+1)$ for some integer $p$.  Even-Zohar \etal\ prove that the expected value of the Casson invariant is $\binom{p}{2}/12$ in the Petaluma model and $\binom{p+1}{3}/2 \approx 0.03 n^{3/2}$ in the star model.

Our defect analysis implies an upper bound on the Casson invariant for knot diagrams generated from \emph{any} family of generic closed curves.

\begin{corollary}
Let $γ$ be any generic closed curve with $n$ vertices, and let $κ$ be a knot diagram obtained by resolving each vertex of $γ$ independently and uniformly at random.  Then $\Abs{\E[c_2(κ)]\strut} = O(n^{3/2})$.
\end{corollary}


Our results also imply that the distribution of the Casson invariant depends strongly on the precise parameters of the random model; even the sign and growth rate of $\E[c_2]$ depend on which curves are used to generate knot diagrams.  For example:
\begin{itemize}
\item 
For random diagrams over the flat torus knot $T({p+1}, p)$, we have $\E[c_2(κ)] = -\binom{p}{3}/4 = -n^{3/2}/24 + Θ(n)$.
\item
For random diagrams over the connected sum $T(p, p+1) \mathbin\# T(p+1, p)$, we have $\E[c_2(κ)] = \left(\binom{p+1}{3}-\binom{p}{3}\right)/4 = \binom{p}{2}/4 = n/16 -  Θ(\sqrt{n})$.
\item
For random diagrams over the connected sum $T(p-1, p) \mathbin\# T(p+1, p)$, we have $\E[c_2(κ)] = 0$.
\end{itemize}
We hope to expand on these initial observations  in a future paper.
\section{Open Problems}

Like Gitler~\cite{g-dtaa-91}, Feo and Provan~\cite{fp-dtert-93}, and Archdeacon \etal~\cite{acgp-frpwg-00}, we conjecture that any $n$-vertex planar graph $G$ can be reduced to a single vertex using only $O(n^{3/2})$ electrical transformations.  Our divide-and-proof of Theorem \ref{Th:defect} suggests a natural divide-and-conquer strategy: Find a useful vertex-cut (a cycle in the radial graph), reduce the graph on one side of the cut \emph{treating the cut vertices as terminals}, and then recursively reduce the remaining graph.  If the reduction on one side of the cut can be carried out in $O(nD)$ steps, perhaps using a variant of Feo and Provan's algorithm~\cite{fp-dtert-93}, then the algorithm would reduce $G$ in $O(n^{3/2})$ time.  Unfortunately, the fastest algorithms currently known for reducing planar graphs with arbitrarily many terminals on the outer face---otherwise known as \emph{circular} planar graphs~\cite{cgv-rep-96,cim-cpgrn-98,cm-ipen-00,k-lpggs-11}---only imply an overall running time of $O(n^2)$.

Because our lower bound applies to any planar graph with treewidth $\Omega(\sqrt{n})$, it is natural to conjecture that every planar graph with treewidth $t$ can be electrically reduced in at most $O(nt)$ steps.  Of course such an algorithm would immediately imply an $O(n^{3/2})$-time algorithm for arbitrary planar graphs.

Another interesting open question is whether our $Ω(n^{3/2})$ lower bound can be extended to electrical transformations that ignore the planarity of the graph.  Such an extension would imply that reducing $K_5$- or $K_{3,3}$-minor-free graphs would require $Ω(n^{3/2})$ steps in the worst case.  As we mentioned in the introduction, our argument can be extended to allow non-facial loop reductions and non-facial parallel reductions; the difficulty lies entirely with non-facial $\arc{Δ}{Y}$ transformations.

There are also several natural open questions involving combinatorial homotopy.  For example: How many homotopy moves are required in the worst case to transform an arbitrary $n$-vertex 4-regular plane graph into a collection of disjoint simple closed curves?  Equivalently, how many moves are required to transform any generic \emph{immersion} of  circles into an \emph{embedding}?  Our $Ω(n^{3/2})$ lower bound for homotopy moves trivially extends to this more general problem; however, an $O(n^{3/2})$-step electrical reduction algorithm would \emph{not} immediately imply the corresponding upper bound for homotopy moves of arbitrary immersions, because of a subtlety in the definition of a planar embedding of disconnected graphs.  In the homotopy problem, if the image of an immersion becomes disconnected, we must still keep track of how the various components are nested; on the other hand, in the electrical reduction problem, it is most natural to embed each component on its own sphere.  If we insist on keeping everything embedded on the same plane, there are disconnected 4-regular plane graphs that \emph{cannot} be reduced to disjoint circles by medial electrical moves; see Figure \ref{F:no-me-reduction}.

\begin{figure}[htb]
\centering
\includegraphics[scale=0.4]{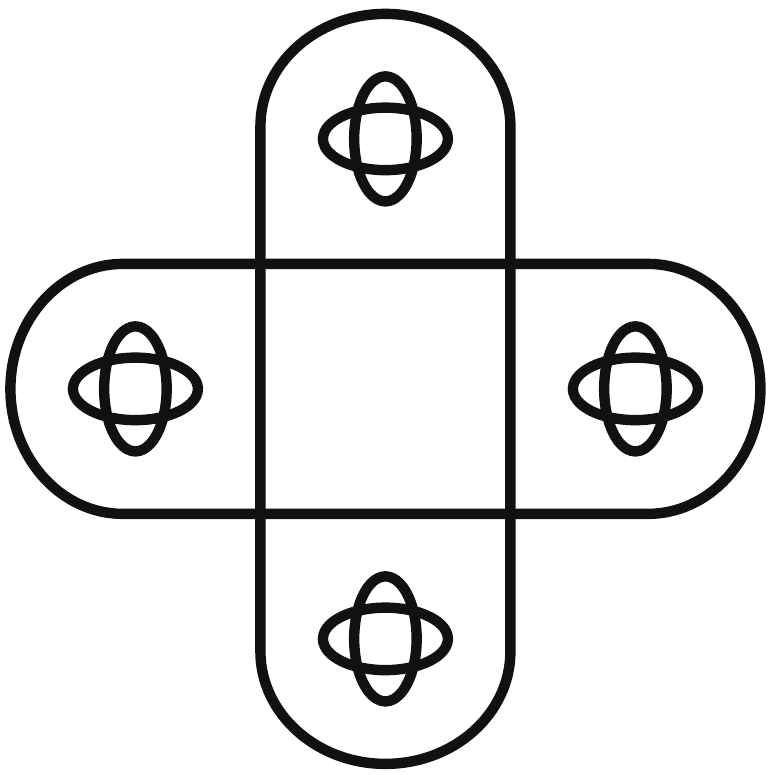}
\caption{A disconnected 4-regular plane graph that cannot be reduced to disjoint circles by medial electrical moves.}
\label{F:no-me-reduction}
\end{figure}

How many homotopy moves are required to connect two homotopic closed curves on a surface of higher genus?  In particular, how many homotopy moves are required to transform a \emph{contractible} closed curve into a simple closed curve?  Polyak's formula for defect extends immediately to curves on arbitrary orientable surfaces, even though other invariants like winding number do not.  There are curves on the torus with quadratic defect, but we have been unable to construct any such curve that is contractible.  For the more general problem of transforming one arbitrary curve into another, we are not even aware of a polynomial upper bound!

Finally, our upper and lower bounds for worst-case defect differ by roughly a factor of $20$.  In light of the complexity of our upper bound argument, we expect that our lower bounds are closer to the correct answer.  Are the flat torus knots $T(q+1, q)$ and $T(p, p+1)$ the curves with minimum and maximum defect for their number of vertices?

\paragraph{Acknowledgements.}
We would like to thank Bojan Mohar and Nathan Dunfield for encouragement and helpful discussions.

\bibliographystyle{newuser}
\bibliography{bib/optimization,bib/topology,bib/data-structures,bib/jeffe,bib/compgeom}

\end{document}